\documentclass[english]{article}
\usepackage[T1]{fontenc}
\usepackage[latin9]{inputenc}
\usepackage{geometry}
\usepackage{babel}
\usepackage{float}
\usepackage{amsmath}
\usepackage{amsthm}
\usepackage{amssymb}
\usepackage{graphicx}
\usepackage[numbers]{natbib}
\usepackage[pdfusetitle,
 bookmarks=true,bookmarksnumbered=false,bookmarksopen=false,
 breaklinks=false,pdfborder={0 0 1},backref=false,colorlinks=false]
 {hyperref}

\makeatletter
\theoremstyle{definition}
\newtheorem*{defn*}{\protect\definitionname}
\theoremstyle{plain}
\newtheorem*{conjecture*}{\protect\conjecturename}
\theoremstyle{plain}
\newtheorem{prop}{\protect\propositionname}
\theoremstyle{plain}
\newtheorem{thm}{\protect\theoremname}
\theoremstyle{plain}
\newtheorem{cor}{\protect\corollaryname}
\theoremstyle{plain}
\newtheorem{lem}{\protect\lemmaname}

\@ifundefined{date}{}{\date{}}
\usepackage{braket}
\usepackage[section]{placeins}
\usepackage{hyperref}

\makeatother

\providecommand{\conjecturename}{Conjecture}
\providecommand{\corollaryname}{Corollary}
\providecommand{\definitionname}{Definition}
\providecommand{\lemmaname}{Lemma}
\providecommand{\propositionname}{Proposition}
\providecommand{\theoremname}{Theorem}

\begin{document}
\title{Neural Network Ground State from the Neural Tangent Kernel Perspective:
The Sign Bias}
\author{\textbf{$\text{Harel Kol-Namer}$, $\text{Moshe Goldstein}$}\\
$\text{School of Physics and Astronomy, Tel-Aviv University, Tel Aviv 6997801, Israel}$}
\maketitle
\begin{abstract}
Neural networks has recently attracted much interest as useful representations
of quantum many body ground states, which might help address the infamous
sign problem. Most attention was directed at their representability
properties, while possible limitations on finding the desired optimal
state have not been suitably explored. By leveraging well-established
results applicable in the context of infinite width, specifically
regarding the renowned neural tangent kernel and conjugate kernel,
a comprehensive analysis of the convergence and initialization characteristics
of the method is conducted. We reveal the dependence of these characteristics
on the interplay among these kernels, the Hamiltonian, and the basis
used for its representation. We introduce and motivate novel performance
metrics and explore the condition for their optimization. By leveraging
these findings, we elucidate a substantial dependence of the effectiveness
of this approach on the selected basis, demonstrating that so-called
``stoquastic\textquotedblright{} Hamiltonians are more amenable to
solution through neural networks than those suffering from a sign
problem.
\end{abstract}

\section*{I. Introduction\protect\label{sec:I}}

Finding the ground state of many-body local spins systems is a notoriously
hard task, being QMA complete under pretty mild conditions \citep{https://doi.org/10.48550/arxiv.quant-ph/0406180},
implying that even with a universal quantum computer one does not
expect it to be tractable. Nonetheless, it is believed, and proven
for one dimensional gapped systems, that in many physical scenarios
with finite range interactions this task should be achievable \citep{Hastings_2006}.
While there are many common numerical methods to find the ground state
and ground state energy, such as Tensor Networks (TN) or Quantum Monte
Carlo (QMC), each has its own strength and weakness and thus applies
to only certain classes of Hamiltonians; for example, TN are mostly
be used for low dimensional models with a ground state satisfying
an area law \citep{Cirac_2021}, and QMC suffers from the infamous
sign problem, which arises when the Hamiltonian is not stoquastic,
as discussed further below.

In recent years the representation of quantum states using neural
networks, i.e., Neural Quantum States (NQS), have been proposed and
demonstrated success in various tasks \citep{Carleo_2017,Gao_2017,Torlai_2018}.
An unnormalized state $\ket{\psi}$ of $N$ spins is represented by
a neural network if it can written as 
\begin{equation}
\ket{\psi}=\sum_{\sigma}\psi(\sigma,\theta)\ket{\sigma}
\end{equation}
in some fixed tensor product basis $\ket{\sigma}$ with $\psi(\sigma,\theta)\equiv\psi_{\sigma}$
being the neural network output which we assume to be real without
loss of generality\footnote{As mentioned in \citep{Cubitt_2018}, every complex k-local Hamiltonian
problem can be casted as a k+1-local real Hamiltonian problem. Also
in \citep{Luo_2023} it was shown that if the NTK is complex the dynamics
can be separated into two equations, for the real and imaginary parts
of the wave functions, with the effective NTKs being the real and
imaginary parts of the complex one, each with similar characteristics.
For each of those equations our analysis applies. }, given an input $\sigma\in\{-1,1\}^{N}$ and the network parameters
$\theta\in\mathbb{R}^{p}$. The output is then fed into a loss function,
$\mathcal{L}(\ket{\psi}):\mathbb{R}^{2^{N}}\rightarrow\mathbb{R}$,
which is minimized via gradient descent (or a variant thereof), through
a change of the parameters $\theta$. The task at hand will specify
a loss function. In the following we will be interested in the representation
of the ground states of local spins Hamiltonian $H$, hence the chosen
loss function will be the energy expectation value $\braket{E}\equiv\frac{\bra{\psi}H\ket{\psi}}{\braket{\psi|\psi}}$. 

Allegedly avoiding the limitations of TN and QMC, neural quantum states
optimization have been proposed as a promising numerical method for
finding the ground state of many-body spins systems. Numerical studies
have shown results comparable to other state of the art methods while
using various neural networks architectures \citep{Carleo_2017,Liang_2018,Hibat_Allah_2020,Sharir_2020}.
Nonetheless it was empirically observed that even though NQS showed
success for models such as the transverse-field Ising model and other
stoquastic Hamiltonians, for non-stoquastic Hamiltonians the performance
of this method seems to diminish drastically. While several reasons
have been proposed, such as inability to represent states with complex
sign structure, Monte Carlo sampling issues, and rugged loss landscape
\citep{Westerhout_2020,Szabo:2020vk,Park_2022,Bukov_2021}, a more
comprehensive explanation is sorely missing. It was also noted that
whenever possible, a change of basis to one where the Hamiltonian
is stoquastic can enhance the performance, indicating that there is
an ambiguous dependence on the chosen basis \citep{Park_2022}\citep{Bukov_2021}\citep{Nomura_2021}.
Thus, the relationship between the performance of the method and the
specific characteristics of the Hamiltonian or the chosen basis remained
elusive. 

In order to explore those questions we use the highly celebrated infinite
width limit \citep{https://doi.org/10.48550/arxiv.1806.07572,https://doi.org/10.48550/arxiv.1711.00165,https://doi.org/10.48550/arxiv.1810.05148,neal2012bayesian}.
In this limit, neural networks dynamics are often more tractable and
enjoy useful features that will be demonstrated later on. We refer
to this limit as the Neural Tangent Kernel (NTK) limit. Similar approach
have shown to be insightful in the previous works, such as the analysis
of quantum neural networks \citep{Liu_2023,https://doi.org/10.48550/arxiv.2402.08726}
and classical neural networks \citep{Luo_2023} for quantum state
tomography tasks. Similar to our work, geometrical methods have been
utilized for analyzing the training dynamics for ground state optimization
of NQS such as in \citep{Park_2020}, yet using different measures
such as the Fisher information metric. 

The advantage of using the NTK limit in our case is two-fold. We are
able to explore the training dynamics in a general setting and the
influence of the chosen basis, while eliminating other influences
that have been explored, such as the representability capabilities
of the network. We are therefore able to demonstrate that NQS suffer
from an implicit bias in their initialization and dynamics toward
states with a simple sign structure, which we call ``the sign bias''.
Based on this bias and the special structure of stoquastic Hamiltonians
ground states, we are able to explain the hardness in finding the
ground states of non-stoquastic Hamiltonians using common NQS training
schemes, as compared to TNs. 

The study is structured as follows: In Sec. \hyperref[sec:II]{II}
we outline the general framework and definitions. In Sec. \hyperref[sec:III]{III}
we review the NTK limit properties based on previous results, and
rephrase them in the language of quantum mechanical operators. In
Sec. \hyperref[sec:IV]{IV} we present our general results regarding
the interplay between the kernels and the Hamiltonian and its influence
on the dynamics and initialization, and build intuition about the
desired relations between them. We provide criteria for exponential
convergence of the NQS towards the ground state, as well as derive
an upper bound on the learning rate in the discrete time case, which
might be useful in practice. Based on all these we propose metrics
that might predict the success or failure of the method. In Sec. \hyperref[sec:V]{V}
we discuss what we term the sign ``bias'', which might explain why
stoquastic Hamiltonians are more amenable to solution using neural
networks, and why one can enhance the convergence rate by a change
of basis. We also propose metrics that might predict the success or
failure of the method and discuss their optimization. In Sec. \hyperref[sec:VI]{VI}
we provide numerical results and demonstrate how the metrics provided
are correlated with the convergence speed and initial average values.
In Sec. \hyperref[sec:VII]{VII} we briefly address a commonly used
optimization strategy, Stochastic Reconfiguration (SR) as a possible
remedy to the sign bias, and more generally, show how one can analyze
its strength and limitations using the tools developed throughout
the work. Finally, in In Sec. \hyperref[sec:VIII]{VIII} we discuss
how our analysis might translate to more general ansatz and examplify
them on TNs, for which we can show that the sign bias does not exist.
Some technical details are provided in the Appendixes.

\section*{II. General framework and definitions \protect\label{sec:II}}

\subsection*{II.A Neural Networks}

Neural network have become very popular in recent years due to their
numerous applications \citep{https://doi.org/10.48550/arxiv.1506.00019,Alzubaidi_2021}.
In general, a neural network is a specific case of a parametric function
$f\left(x,\theta_{i}\right)$, with $x$ being the input to the network,
which might be images of hand written letters, time series, or almost
any other object which can be represented digitally. $\theta_{i}$
with $i\in\left\{ 1,...,p\right\} $ are the $p$ variational parameters
of the network; with different parameters the network represents different
functions. 

Based on the task at hand, one usually defines a loss function $\mathcal{L}\equiv\mathcal{L}\left(f\left(x,\theta\right)\right)$,
which is a functional of the parametric function, and minimizes it
by varying the parameters using various gradient descent methods.
The most basic gradient descent has the following update rule 
\begin{equation}
\theta_{i}^{n+1}=\theta_{i}^{n}-\Delta t\frac{\partial\mathcal{L}}{\partial\theta_{i}},\label{eq:2}
\end{equation}
with the gradient evaluated at step $n$ and $\Delta t$ being referred
to as the learning rate. The main characteristic of neural networks
that distinguish them from other parametric functions is their special
recursive form 
\begin{equation}
f\left(x,\theta_{i}\right)=h_{L}\left(h_{L-1}\left(h_{L-2}\left(...\right),\theta_{L-1}\right),\theta_{L}\right),\label{eq:3}
\end{equation}
with $L$ being an integer, usually referred to as the number of layers,
and $h_{l}$ being the output of the $l$th layer, $l=1.\cdots,L$,
which depends on the subset of parameters $\theta_{l}$ and the previous
output $h_{l-1}$. This structure is what allows for neural networks
to be very deep, and thus to represent complicated functions, while
still having tractable gradients \citep{HORNIK1991251}.

In this work we will consider a feedforward neural network with a
general number of layers $L$ and an activation function $\phi$,
as can be seen in Fig. \ref{fig:A-general-fully-1}, where 
\begin{equation}
h_{l}=\phi\left(W_{l}h_{l-1}+b_{l}\right),
\end{equation}
i.e., each layer is composed of an activation function acting element-wise
on the results of an affine transformation specified by $W_{l}\in\mathbb{R}^{p_{l}\times p_{l-1}}$
and $b_{l}\in\mathbb{R}^{p_{l}}$, with $p_{l}$ being referred to
as the width of the $l$th layer. Each affine transformation acts
on the previous hidden units $h_{l-1}\in\mathbb{R}^{p_{l-1}}$, with
$h_{0}\equiv x\in\mathbb{R}^{n}$ being the input of the network.
We group all of the trainable parameters of the networks and denote
it using a vector $\theta\in\mathbb{R}^{p}$. 

We will also mostly consider the continuous time limit of the gradient
descent updates in Eq. (\ref{eq:2}), where it turns into a gradient
flow 
\begin{equation}
\frac{d\theta_{i}}{dt}=-\frac{\partial\mathcal{L}}{\partial\theta_{i}}.
\end{equation}
Using the chain rule we get an ODE for the state vector 
\begin{equation}
\frac{d\psi_{\sigma}}{dt}=-\sum_{\sigma'}\Theta_{\sigma\sigma'}^{t}\frac{\partial\mathcal{L}}{\partial\psi_{\sigma}},\label{eq:6}
\end{equation}
with $\Theta_{\sigma\sigma'}^{t}=\sum_{i}\frac{\partial\psi_{\sigma}}{\partial\theta_{i}}\frac{\partial\psi_{\sigma'}}{\partial\theta_{i}}$.
The time dependence of $\Theta_{\sigma\sigma'}^{t}$ makes the analysis
of this equation intractable, since the dynamics still depends on
the parameters values. Fortunately, recent works in neural networks
analysis \citep{https://doi.org/10.48550/arxiv.1806.07572,https://doi.org/10.48550/arxiv.1711.00165,https://doi.org/10.48550/arxiv.1810.05148,neal2012bayesian}
have shown that when one takes the infinite width limit, i.e., the
number of hidden variables in each layer grows to infinity, several
useful properties emerge. Their applicability for the energy loss
function which we consider here is shown in \hyperref[subsec:App.A]{App. A}. 

The first property is that $\Theta$, which will be referred to as
the NTK in this context, is time independent and completely determined
by the initialization and architecture of the network. A second property
is that a random initialization of the parameters translates into
a normal distribution of the elements of the initial state vector,
i.e.,
\begin{equation}
\langle\sigma|\psi^{0}\rangle\sim\mathcal{N}(\mu,K),
\end{equation}
with $\ket{\psi^{0}}$ being the initial state, and where $K$ is
referred to as the Conjugate Kernel (CK), which is determined by the
initialization and network architecture. We will consider the case
where $\mu=0$, which corresponds to zero bias initialization. 

In our case of interest, andassuming a unique ground state of $H$
(extension to degenerate ground states are possible), the loss function
will be the energy loss function 
\begin{equation}
\mathcal{L}=\braket{E}\equiv\frac{\braket{\psi\lvert H\lvert\psi}}{\braket{\psi\vert\psi}}.
\end{equation}
Hence, using Eq. (\ref{eq:6}) we get the ODE for the state
\begin{equation}
\frac{d\ket{\psi}}{dt}=\frac{1}{\braket{\psi|\psi}}\Theta\left(\braket{E}-H\right)\ket{\psi},\label{eq:8}
\end{equation}
which will be our main interest throughout the paper. We note that
the non-triviality of the kernels is responsible for the basis dependence
of the dynamics and initialization. In the case of a trivial NTK,
the equation reduces to the imaginary time Schrödinger equation, and
would be invariant under basis transformations.

Throughout the paper we disregard the issues which may arise from
the use of Monte-Carlo sampling, and since we address the infinite
width limit, we can disregard any influence of the representability
of the network \citep{HORNIK1991251,https://doi.org/10.48550/arxiv.1806.07572}.
These settings thus allow us to eliminate other factors which might
influence the capability of neural network in solving the many-body
problem and focus on the trainability and dynamics. This in turn will
allow us to elucidate the fundamental interplay between the form of
the Hamiltonian in a given basis and the architecture of the NQS.

\begin{figure}[tb]
\centering{}\includegraphics[scale=0.5]{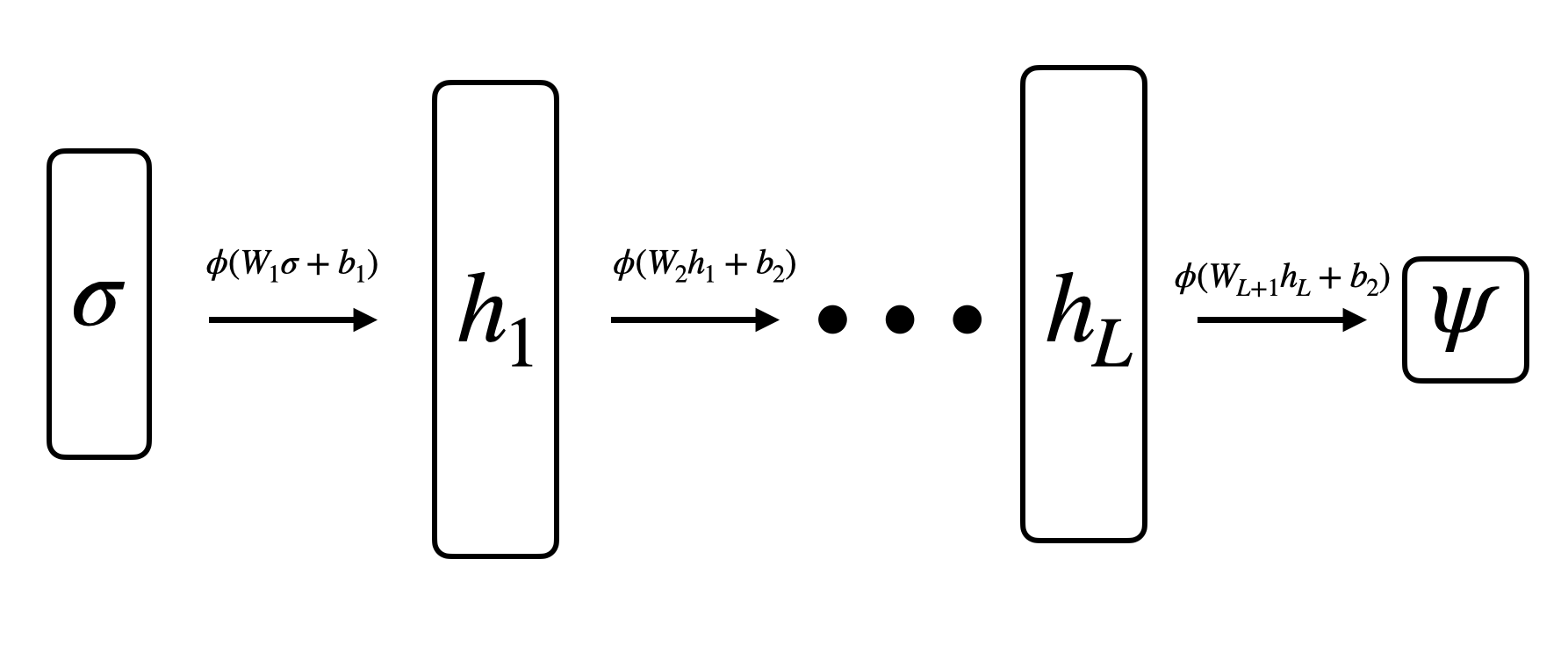}\caption{\protect\label{fig:A-general-fully-1}A general fully connected architecture
for NQS. The inputs are $\sigma\in\left\{ -1,1\right\} ^{N}$, each
hidden layer $h_{l}$ with $l=1,...,L$ represents the output of the
affine transformation and non-linear activation, $\phi\left(W_{l}h_{l-1}+b_{l}\right)$,
over the previous layer. The final output is the quantum amplitude
of the corresponding spin configuration $\psi\left(\sigma\right)$.}
\end{figure}

\subsection*{II.B Many-Body Spin Hamiltonians }

Many-body spin Hamiltonians are the subject of extensive study in
fields such as condensed matter physics, high-energy physics, and
quantum computation \citep{Albash_2018,Logan_2005}. In our case one
can always write them as 
\begin{equation}
H=\sum_{i}^{r}h_{i},
\end{equation}
with each $h_{i}$ being a tensor product of Pauli matrices, which
we denote by $X,Y,Z$, and identity matrices. Usually one is concerned
with Hamiltonians which have a certain local structure, so we will
mostly consider what is known as $k$-local Hamiltonian. For those
Hamiltonians each $h_{i}$ acts non-trivially on at most $k$ spins
or qubits. For example, the Transverse Field Ising Model (TFIM)
\begin{equation}
H_{\text{{TFIM}}}=-\sum_{i}hX_{i}-\sum_{\braket{i,j}}Z_{i}Z_{j},\label{eq:13}
\end{equation}
is a 2-local Hamiltonian model, with $\braket{i,j}$ indicates pairs
of nearest neighbors over a general interaction graph. 

A known class of Hamiltonians that was first brought up in the context
of QMC methods, are stoquastic Hamiltonians:
\begin{defn*}
\emph{Stoquastic Hamiltonian} - A Hamiltonian is said to be stoquastic
in a given tensor-product basis $\ket{\sigma}$ if it satisfies $\braket{\sigma|H|\sigma'}\leq0$,
$\forall\sigma\neq\sigma'$. 
\end{defn*}
An example is the TFIM. Stoquastic Hamiltonians have been an object
of interest due to their many special properties and applications.
Stoquastic Hamiltonian always have non-negative ground state amplitudes
$\braket{\sigma|g}\geq0$, and a Gibbs operator with non-negative
matrix elements, $\braket{\sigma|e^{-\beta H}|\sigma'}\geq0$, $\forall\beta$
in the given basis \citep{Klassen_2019}. Those properties are crucial
for the applicability of various QMC methods; and in most cases when
the Hamiltonian is non-stoquastic, QMC methods will encounter the
notorious sign problem \citep{Austin_2011}. Since the definition
of stoquastic Hamiltonian depends on the basis, in some cases one
can attempt to find a different tensor-product basis where the Hamiltonian
is stoquastic, although in general this task is NP-hard \citep{https://doi.org/10.48550/arxiv.1906.08800}\citep{Marvian_2019}. 

\section*{III. Kernel Properties\protect\label{sec:III}}

We first review previously known properties about the NTK and CK in
a more quantum mechanical language, following \citep{https://doi.org/10.48550/arxiv.1806.07572,https://doi.org/10.48550/arxiv.1907.10599}:
\begin{enumerate}
\item Both are positive definite matrices, i.e. $\braket{\psi|\Theta|\psi}>0,\ \braket{\psi|K|\psi}>0,\ \forall\ket{\psi}$
\item Their matrix elements are of the form $\bra{\sigma}\Theta\ket{\sigma'}=\Phi(\sigma\cdot\sigma'),\ \bra{\sigma}K\ket{\sigma'}=F(\sigma\cdot\sigma')$,
with $\sigma\in\{-1,1\}^{N}$ and $\Phi,\ F$ being functions which
depends on the number of layers and activation function that fully
characterize the kernels.
\item We define by $X_{\vert S\vert}\equiv\sum_{\mathcal{\vert S\vert}=\vert S\vert}X_{\mathcal{S}}$
the sum of all products $X_{\mathcal{S}}=\prod_{i\in S}X_{i}$ with
$\mathcal{S}\subset\{1,...,N\}$, acting on precisely $\vert S\vert$
spins. By the previous property, one can write the kernels in the
form $\Theta=\sum_{\lvert S\lvert\leq N}\alpha_{\lvert S\lvert}X_{\lvert S\lvert}$,
with $\alpha_{\lvert S\lvert}=\Phi(N-2\lvert S\lvert)$, and similarly
for the CK. 
\item It immediately follows from the previous property that the eigenstates
of the kernels are the vectors $\ket{s}$ with $s\in\left\{ -,+\right\} ^{N}$,
which are the product states of the $X$ Pauli matrix eigenstates.
The eigenvalue associated with each $\ket{s}$ depends only on $\lvert s\lvert$,
which we define to be the number of $s_{i}=-$ in $s$. Thus, the
spectrum of both kernels breaks down to $N+1$ degenerate subspaces,
each with dimension ${N \choose \lvert s\lvert}$. We denote the eigenvalue
associated with a sector $\lvert s\lvert$ by $\Theta_{\lvert s\lvert}$
and $K_{\lvert s\lvert}$ for the NTK and CK, respectively.
\item ``Weak Bias'' -- The eigenvalues are always ordered as follows
\begin{equation}
\Theta_{0}\geq\Theta_{2}\geq\cdots\geq\Theta_{2k}\geq\cdots,\label{eq:14}
\end{equation}
\[
\Theta_{1}\geq\Theta_{3}\geq\cdots\geq\Theta_{2k+1}\geq\cdots\ ,
\]
and similarly for the CK. It immediately follows that the maximal
eigenvalue of the kernels is either $\Theta_{0}$ with associated
eigenstate $\ket{+}\equiv\ket{+}^{\otimes N}$ , or $\Theta_{1}$
with associated eigenspace spanned by $\ket{-_{i}}$ with $-_{i}$
denoting a string of all $+$ except one $-$ sign in the $i$th position. 
\item If $\Phi\geq0$ then $\Theta_{0}$ is the maximal eigenvalue, and
similarly for the CK.
\item For the CK, if the activation function $\phi>0$ then $F>0$. We note
that this does not apply in general for the NTK. 
\item Denoting by $\kappa_{k}$ the $k$th unique largest eigenvalue of
$\Theta$, with a normalization such that $\kappa_{1}=1$, the eigenvalues
decay as $\sim\frac{1}{2^{Nk}}$. In the common case where the eigenvalues
are ordered by $\lvert s\lvert$ then $\Theta_{\lvert s\lvert}\sim\frac{1}{2^{N\lvert s\lvert}}$.
\end{enumerate}
We mention that as in \citep{https://doi.org/10.48550/arxiv.1907.10599},
one can associate a measure of ``complexity'' for each eigen-space,
since each one of them is fully characterized by the number $\lvert s\lvert$,
i.e., the number of minus signs in the states $\ket{s}$ which span
it. Since for a single qubit $\ket{+}=\frac{1}{\sqrt{2}}\left(\ket{0}+\ket{1}\right)$
and $\ket{-}=\frac{1}{\sqrt{2}}\left(\ket{0}-\ket{1}\right)$, the
number $\lvert s\lvert$ measure the number of dimensions in which
the state $\ket{s}$ is ``oscillating'' while being constant in
the other $N-\lvert s\lvert$ dimensions. This notion of complexity
stems from the relation to the discrete Fourier Transform, in which
a discrete time signal $x_{n}$ with $n\in\{0,...,N-1\}$ can be decomposed
as $x_{n}=\sum_{k=0}^{N-1}X_{k}e^{i\frac{2\pi}{N}kn}$. If the significant
contributions to the sum comes from large $k$ then the signal is
highly oscillating. 

\section*{IV. Main Results \protect\label{sec:IV}}

We begin by discussing the role of the NTK and CK, presenting basic
results and their implications, with derivations given in \hyperref[subsec:App.B]{App. B}

\subsection*{IV.A The NTK \protect\label{subsec:IV.A}}

In order to understand the role of the NTK we discretize the gradient
flow defined by Eq. (\ref{eq:8}),
\begin{equation}
\ket{\psi_{n+1}}=\text{argmin}_{\ket{\psi}}\left\{ \frac{\bra{\psi}H\ket{\psi}}{\braket{\psi|\psi}}+\frac{1}{2\Delta t}\left(\bra{\psi}-\bra{\psi_{n}}\right)\Theta^{-1}\left(\ket{\psi}-\ket{\psi_{n}}\right)\right\} ,
\end{equation}
with the index $n$ describing the discrete time step $t_{n}$ and
$\Delta t\equiv t_{n+1}-t_{n}$. One can see that this is the usual
gradient flow, yet instead of having the usual Euclidian metric, the
metric is $\Theta^{-1}$, which is the metric induced on the space
of the unnormalized state vectors when having the Euclidian metric
in parameter space. The dynamics can now be interpreted as follows:
at any point $\ket{\psi_{n}}$ the algorithm searches for the minima
of the energy in a small neighborhood defined by the metric $\Theta^{-1}$,
which favors changes along directions with small $\braket{\Theta^{-1}}$.
That bias toward directions with small inverse NTK expectation was
noted in previous works in other contexts, such as regression tasks
\citep{https://doi.org/10.48550/arxiv.1806.07572,https://doi.org/10.48550/arxiv.1907.10599,https://doi.org/10.48550/arxiv.1905.10264}.
The severity of that bias depends on the eigenvalues of the kernel.
Using Property 5 from the previous section, if $\Theta_{0}$ or $\Theta_{1}$
is much larger than any other eigenvalue, the kernel gradient flow
will be highly biased. 

A similar conclusion can be reached by examining the velocity along
the NTK eigenstates. Using Eq. (\ref{eq:8}) we have 
\begin{equation}
\frac{d\braket{s|\psi}}{dt}=\Theta_{\lvert s\lvert}\bra{s}\left(\braket{E}-H\right)\ket{\psi},
\end{equation}
which can be interpreted as having a generalized gradient flow with
a different learning rate in each direction $\ket{s}$, which is proportional
to $\Theta_{\lvert s\lvert}$. Thus the flow will be faster along
directions with larger $\Theta_{\lvert s\lvert}$. 

From the above we formulate the following compatibility conjecture
\begin{conjecture*}
Compatibility: For a Hamiltonian $H$ with a ground state $\ket{g}$,
the metric $\bra{g}\Theta\ket{g}$ is positively correlated with the
ease of finding the ground state of $H$ using ANN with NTK $\Theta$. 
\end{conjecture*}
In order to lend further support to these ideas, let us consider the
specific case that $\left[H,\Theta\right]=0$. Then, denoting the
eigenbasis amplitudes of the normalized state by $\frac{\braket{E_{i}|\psi}}{\sqrt{\braket{\psi|\psi}}}\equiv\tilde{\psi}_{i}$
one can show the following 
\begin{prop}
If $\left[H,\Theta\right]=0$, a necessary condition for 
\[
\frac{d}{dt}\left(\frac{\tilde{\psi_{g}^{2}}}{\tilde{\psi}_{i}^{2}}\right)>0
\]
 is
\[
\frac{\braket{g|\Theta|g}}{\braket{E_{i}|\Theta|E_{i}}}>\frac{\braket{E}-E_{i}}{\braket{E}-E_{g}}.
\]
For $E_{i}<\braket{E}$ the condition is non-trivial. 
\end{prop}
Using the above we get that in order for the ground state component
to increase relative to the other components, the bias induced toward
the ground state relative to other directions cannot be too small.
One way to satisfy the criteria $\forall t$ and every eigenstate
is to have 
\begin{equation}
\braket{g|\Theta|g}\gg\braket{E_{i}|\Theta|E_{i}}\ ,\ \forall i>0.
\end{equation}
For some energy eigenstates or parts of the optimization this criterion
is actually necessary. Consider the case where $\braket{E}\gg E_{i}$,
which might occur for $\ket{E_{i}}$ being a low energy eigenstate
such as the first excited state and $t\rightarrow0$, or when we have
a relatively small energy gap between the ground state and the first
excited state. In that case the necessary condition translates into
\begin{equation}
\frac{\braket{g|\Theta|g}}{\braket{E_{i}|\Theta|E_{i}}}\gtrsim1,
\end{equation}
which means that the bias induced toward the ground state must be
larger than the bias toward other low energy eigenstates. Otherwise,
the energy might decrease, yet due to the NQS becoming more similar
to a low energy excited state rather than to the ground state. Only
when $\braket{E}\approx E_{i}$ after significant time elapsed, will
the NQS start to become more similar to the ground state. Thus, it
is preferable that the NTK expectation over the ground state will
be much larger than the expectation over the other energy eigenstates,
and in some cases this is even necessary in order to achieve reasonable
convergence times.

Given some value of $\bra{g}\Theta\ket{g}$, one should ask how the
NTK expectation values for the other energy eigenstates influence
the dynamics. This can be understood from the solution for Eq. (\ref{eq:8})
near the ground state, i.e. $\ket{\psi}=\ket{g}+\epsilon\ket{\phi}$
with $\epsilon\ll1$ and the initial perturbation $\ket{\phi_{0}}$
satisfying$\braket{g|\phi_{0}}=0$ and $\braket{\phi_{0}|\phi_{0}}=1$,
\begin{equation}
\ket{\psi}=\ket{g}+\epsilon\sqrt{\Theta}e^{-\sqrt{\Theta}\left[H-E_{g}\right]\sqrt{\Theta}t}\sqrt{\Theta^{-1}}\ket{\phi_{0}},
\end{equation}
where we assumed without loss of generality that the norm at $t\rightarrow\infty$
is 1, see \hyperref[subsec:App.B]{App. B} for further details. From
the solution one can see that the convergence to the ground state
will be dictated by the eigenvalues of $\tilde{H}=\sqrt{\Theta}\left[H-E_{g}\right]\sqrt{\Theta}$.
This matrix is positive definite, so its ground state is $\sqrt{\Theta^{-1}}\ket{g}$
with eigenvalue 0. Denoting the eigenvalues of $\tilde{H}$ by $\alpha_{i}$
\begin{equation}
\epsilon^{2}\alpha_{1}\bra{\phi_{0}}\sqrt{\Theta^{-1}}\ket{\alpha}^{2}e^{-2\alpha_{1}t}\leq\lvert\braket{E}-E_{g}\lvert\leq\epsilon^{2}\bra{\phi_{0}}\left[H-E_{g}\right]\ket{\phi_{0}}\alpha e^{-2\alpha_{1}t}
\end{equation}
with $\ket{\alpha_{1}}$ being the first excited state of $\tilde{H}$
with eigenvalue $\alpha_{1}$. Thus, the energy converges to the ground
state as $\sim e^{-2\alpha_{1}t}$, and so the larger $\alpha_{1}$
the fastest the convergence will be. 

We are interested in maximizing $\alpha_{1}$ given an Hamiltonian.
As we show in \hyperref[subsec:App.B]{App. B}, it can be seen that
this maximum occurs when $\alpha_{i}=C$ is independent of $i$ for
$i>0$. We also show that it means that $\braket{\Theta}_{\ket{E_{i}}}\propto\frac{1}{E_{i}-E_{g}}$
for $i>0$. This result provides an indication that a desirable relation
between the Hamiltonian and the NTK is such that their spectra are
inversely correlated, such that increasing the NTK expectation means
decreasing the energy. 

The same insight also arises by considering Eq. (\ref{eq:8}) for
the specific case $\left[H,\Theta\right]=0$ and noticing that $\frac{d\psi_{i}^{2}}{dt}\propto\braket{E_{i}|\Theta|E_{i}}\left(\braket{E}-E_{i}\right)\psi_{i}^{2}$.
Thus, $\psi_{i}$ decays for $E_{i}>\braket{E}$ and increases for
$E_{i}<\braket{E}$. Considering low $\braket{E}$ close to $E_{g}$,
high energy components with $E_{i}\gg\braket{E}$ will decay fast
due to the proportionality with $\braket{E}-E_{i}$. On the other
hand, for low energy components with $E_{i}\gtrsim\braket{E}$, one
will need $\braket{E_{i}|\Theta|E_{i}}$ to be large in order to compensate
for the small $\braket{E}-E_{i}$. Since having a large $\braket{E_{i}|\Theta|E_{i}}$
for low energy eigenstates will also increase the low energy components
of the state for large $\braket{E}\gg E_{g}$, we suspect that large
bias toward low energy states in general will cause fast energy decay
for the entirety of the optimization. 

Combining this idea with the previous Compatibility Conjecture we
formulate a Strong Compatibility Conjecture 
\begin{conjecture*}
Strong Compatibility: For a Hamiltonian $H$ with eigenstates $\ket{E_{i}}$,
the convergence rate toward the minimal energy is positively correlated
with whether or not the following condition is met
\[
\braket{\Theta}_{\ket{g}}\gg\braket{\Theta}_{\ket{E_{1}}}>\cdots>\braket{\Theta}_{\ket{E_{k}}}>\cdots,
\]
i.e., if the task of decreasing the energy is compatible with the
task of increasing the NTK expectation.
\end{conjecture*}

\subsection*{IV.B The CK \protect\label{subsec:IV.B}}

The role of the CK is at the initialization, since it is the covariance
matrix of the Gaussian distribution of the initial state. One can
see that for a general $N$ qubits normalized state $\ket{x}$ 
\begin{equation}
\mathbb{E}\left(\frac{\braket{\psi|x}^{2}}{\braket{\psi|\psi}}_{t=0}\right)\approx\frac{\bra{x}K\ket{x}}{\text{Tr}(K)},
\end{equation}
meaning that the initial overlap with the ground state is $\frac{\bra{g}K\ket{g}}{\text{Tr}(K)}$
on average, and the approximation should be valid for a large number
of qubits, see \hyperref[subsec:App.B]{App. B} for the argument.
Thus, in general one would like a high expectation value of the CK
over the ground state in order for the initial state to be close to
the ground state. If either of the eigenvalues $K_{0}$ or $K_{1}$
is much larger than any other value, then the initialization will
be highly biased; for example if $K_{0}\gg K_{i}\ \forall i>0$, then
$\mathbb{E}\left(\frac{\braket{\psi|g}^{2}}{\braket{\psi|\psi}}\right)\approx\braket{g|+}^{2}$. 

The average initial energy can also be seen to be 
\begin{equation}
\mathbb{E}\left(\braket{E}_{t=0}\right)\approx\frac{\text{Tr}\left(KH\right)}{\text{Tr}\left(K\right)}.
\end{equation}
Thus, again, in order to initialize with a low energy, lower energy
states should have high CK expectation values, similarly to the desired
relation between the NTK and the Hamiltonian. We note that one can
calculate the average initial values of higher moments by using Wick's
theorem due to the Gaussian nature of the initial distribution.

\subsection*{IV.C Convergence \protect\label{subsec:IV.C}}

In this section we wish to state general results about the convergence
properties of the gradient descent dynamics. We begin with presenting
a general result about the convergence of the dynamical system.  
\begin{thm}
The continuous time system described by Eq. (\ref{eq:8}) will always
converge to some state. The fixed points of the system are any state
that is proportional to an eigenstate of $H$, with the only stable
point being the one proportional to the ground state. 
\end{thm}
In the following theorem we provide different bounds on the rate of
convergence, which, under reasonable conditions, is exponential with
time. 
\begin{thm}
Given a Hamiltonian $H$ with eigenstates $\ket{E_{i}}$, eigenvalues
$E_{i}$, a ground state $\ket{g}$ with energy $E_{g}$ and a gap
$\Delta$. Denoting the energy by $\braket{E}$, $\braket{E_{0}}$
as the initial energy, and the ground state overlap by $\tilde{\psi}_{g}$,
then for every $C>2$\\
\textbf{Case (1)}: If $\tilde{\psi}_{g}\neq0\ \forall t$ then
\[
\lvert\braket{E}-E_{g}\lvert\leq\begin{cases}
\braket{E_{0}}-E_{g}\lvert e^{-a\Delta t}, & \lvert\braket{E}-E_{g}\lvert\geq\frac{\Delta}{C},\\
\braket{E_{0}}\frac{1}{e^{tb}+\lvert\braket{E_{0}}-E_{g}\lvert\frac{a}{2b}\left(e^{tb}-1\right)}, & \lvert\braket{E}-E_{g}\lvert\leq\frac{\Delta}{C},
\end{cases}
\]
with $a=\frac{\min\left(\tilde{\psi}_{g}^{2}\right)}{\lvert\lvert\Theta\lvert\lvert C}\frac{1}{\braket{g|\Theta^{-1}|g}}$,
$b=\frac{C-1}{C^{2}\Vert\Theta\Vert}\max\left(\frac{1}{2^{N+1}-2}\min_{i\geq1}\left(\frac{\lvert E_{i}-\Delta/C\lvert}{\braket{i|\Theta^{-1}|i}}\right),\ \Delta\min_{i}\left(\Theta_{i}\right)\right)$,
where $\min\left(\tilde{\psi}_{g}^{2}\right)$ is a minimum over the
entire time of the optimization. Here $\Vert\cdot\Vert$ denotes the
operator norm. \\
\textbf{Case (2)}: If $\ket{g}$ is an eigenstate of $\Theta$ then
\[
\lvert\braket{E}-E_{g}\lvert\leq Ae^{-\alpha\Delta t}\ \ \text{for}\ \ \lvert\braket{E}-E_{g}\lvert\geq\frac{\Delta}{C},
\]
with $\alpha=\frac{2\braket{g|\Theta|g}}{\Vert\Theta\Vert C}$, $A=\frac{\Vert\Theta\Vert^{2}}{\Theta_{\min}^{2}}\frac{\Vert H\Vert-E_{g}}{\tilde{\psi}_{g}^{2}\left(t=0\right)}$
and, $\Theta_{\min}^{2}$ being the smallest eigenvalue of $\Theta$.
For $\lvert\braket{E}-E_{g}\lvert\leq\frac{\Delta}{C}$ the bound
is the same as in (1).\\
\textbf{Case (3)}: If $\left[H,\Theta\right]=0$ then
\[
\lvert\braket{E}-E_{g}\lvert\leq\lvert\braket{E_{0}}-E_{g}\lvert\frac{1}{e^{tb}+\lvert\braket{E_{0}}-E_{g}\lvert\frac{a}{2b}\left(e^{tb}-1\right)}\ \ \text{for}\ \ \lvert\braket{E}-E_{g}\lvert\leq\frac{\Delta}{C},
\]
with $a=\frac{\Theta_{\min}}{\Vert\Theta\Vert C}\tilde{\psi}_{g}^{2}\left(t=0\right)\braket{g|\Theta|g}$
and $b=\frac{C-1}{C^{2}\Vert\Theta\Vert}\min_{i\geq1}\left(\braket{i|\Theta|i}\left(E_{i}-\frac{\Delta}{C}\right)\right)$.
For $\lvert\braket{E}-E_{g}\lvert\geq\frac{\Delta}{C}$ the bound
is the same as in (2).
\end{thm}
Theorem 2 states that in the most general case, as long as the network
optimization process does not dramatically fail (vanishing overlap)
then the convergence is exponential with time and dramatically depends
on the relationship between the Hamiltonian and the NTK. One can see
that the more $\ket{g}$ points in the direction of the principal
component of $\Theta$, and the more the other eigenvalues of the
NTK are inversely related to the energy spectrum, the faster the bounds
will decay to zero with time.

Cases (2) and (3) deal with the specific cases when the ground state
and/or other states are also eigenstates of the NTK. In those cases
the convergence is guaranteed and the convergence rate for the bounds
are better. Using what we know from Sec. \hyperref[sec:III]{III},
the third case holds for every Hamiltonian of the form $\sum_{S}\alpha_{S}X_{S}$
with $S\subset\left\{ 1,...,N\right\} $ and $\alpha_{S}$ being real. 

It is also important to note that, in accordance with Sec. \hyperref[subsec:IV.A]{IV.A},
the bounds change below and above the gap. While in the imaginary
time Schrödinger equation the convergence scales as $\sim e^{-at}$
with $a=\mathcal{O}\left(\Delta\right)$, in the non-trivial NTK case
we have $a=\mathcal{O}\left(\braket{g|\Theta|g}\Delta\right)$ above
the gap, depending on the ground state relation with the NTK, while
below the gap $a$ is related to an effective gap which depends on
the relation between the higher energy eigenstates and the NTK.

It turns out that instead of requiring that $\ket{g}$ will be a ground
state of the NTK to achieve better bounds and guarantee convergence,
it is enough to require that the variance of the NTK over the ground
state will be sufficiently small.
\begin{prop}
If the inequality
\[
\frac{1}{4}\frac{\braket{g|\Theta|g}^{2}}{\braket{g|\Theta^{2}|g}-\braket{g|\Theta|g}^{2}}\geq\frac{\lvert\lvert H\lvert\lvert}{\Delta}\frac{\Vert\Theta\Vert C}{\psi_{g}^{2}\left(0\right)}
\]
holds, then case (2) in Theorem 2 holds. 
\end{prop}
One can see that the L.H.S goes to infinity whenever the ground state
is an eigenstate of the NTK. The R.H.S will get smaller if the initial
overlap with the ground state is larger, which on average depends
on the CK expectation value over the ground state. Thus, again, the
desirable relation between both kernels and the Hamiltonian is that
lower energy states will have larger kernel expectation values. 

\subsection*{IV.D Learning Rate \protect\label{subsec:IV.D}}

The learning rate is a fundamental hyper-parameter in the context
of machine learning. Setting the right learning rate is crucial for
convergence. Choosing a learning rate too small will make the convergence
very slow, yet choosing a learning rate that is too high will eventually
prevent successful convergence. Thus, one usually wishes to choose
the largest learning rate while still converging. We provide an upper
bound on the learning rate that must be satisfied in order for the
process to converge.
\begin{thm}
A necessary condition for the convergence to the ground state is 
\[
\Delta t\leq\frac{\braket{\psi_{0}|\Theta^{-1}|\psi_{0}}}{\braket{g|\Theta^{-1}|g}}\frac{2}{\left\Vert \sqrt{\Theta}\left(H-E_{g}\right)\sqrt{\Theta}\right\Vert }.
\]
\end{thm}
Although this upper bound cannot be calculated easily, one can see
that it is inversely proportional to $\braket{g|\Theta^{-1}|g}$,
so the more the ground state points in the direction of the principal
component of the NTK, the larger the learning rate can be. The operator
norm is a bit harder to asses, yet in the highly biased case, when
the maximal eigenvalue is much larger than all other eigenvalues,
one can write the normalized NTK as $\Theta\approx\ket{\phi}\bra{\phi}$,
with $\ket{\phi}$ being $\ket{\Theta_{0}}$ or $\ket{\Theta_{1}}$
by the results quoted in Sec. \hyperref[sec:III]{III}. It follows
that $\left\Vert \sqrt{\Theta}\left(H-E_{g}\right)\sqrt{\Theta}\right\Vert \approx\braket{E}_{\ket{\phi}}-E_{g}$.
Hence, the bound will be higher if the energy expectation value over
the principal component of the NTK will be lower, which echos the
same inverse relation between the Hamiltonian and the kernels that
is postulated to control the overall performance of the method. 

\section*{V. The Sign Bias and basis Dependence \protect\label{sec:V}}

\subsection*{V.A Performance Metrics \protect\label{subsec:V.A}}

In light of the above, we define some useful metrics which quantify
the compatibility between the Hamiltonian in a given basis and a kernel.
We begin with metrics that relate to the bias in the dynamics and
initializations toward the ground state,
\begin{equation}
\mathcal{M}\left(H,\Theta\right)\equiv\frac{\braket{g|\Theta|g}}{\text{\ensuremath{\Vert\Theta\Vert}}},\ \mathcal{M}\left(H,K\right)\equiv\frac{\braket{g|K|g}}{\text{Tr}\left(K\right)}.
\end{equation}
The metric $\mathcal{M}\left(H,\Theta\right)$ is the dynamical bias
toward the ground state, thus providing a measure which is correlated
with the convergence rate towards it. $\mathcal{M}\left(H,K\right)$
is a rough approximation for the average initial ground state overlap,
and should be positively correlated with it. In both cases larger
values are desirable. 

Since the relations between other energy eigenstates and the kernels
are also important for the overall performance, we also consider the
addtitional metrics 
\begin{equation}
\mathcal{N}\left(H,\Theta\right)\equiv\text{Tr}\left(H\Theta\right)-\text{Tr}\left(H\right)\text{Tr}\left(\Theta\right),\ \mathcal{N}\left(H,K\right)\equiv\frac{\text{Tr}\left(\left(H-E_{g}\right)K\right)}{\text{Tr\ensuremath{\left(K\right)}}}.
\end{equation}
$\mathcal{N}\left(H,K\right)$ is an approximation for the average
initial energy, and should be correlated with it. $\mathcal{N}\left(H,\Theta\right)$
is a possible metric for the overall relationship between the Hamiltonian
and the NTK, and quantifies how much the desirable inverse correlation
between them holds. We thus associate it with the convergence time
of the energy to the ground state energy, such that the lower its
value the lower the time it takes to reach it. We note that the second
term in $\mathcal{N}\left(H,\Theta\right)$ is present to account
for the invariance of the dynamics under the addition of a multiple
of the identity to the Hamiltonian.

We further note that while the CK metrics are derived from an approximation
to the average initial values of the ground state overlap and energy,
the NTK metrics are less fundamental. We use the NTK metrics to compare
different bases while other parameters are kept constant. In other
scenarios one might want to use different NTK metrics, which incorporate
other factors. For example, in order to distinguish between different
models and architectures one will also need to incorporate the dynamics
dependence on quantities such as the Hamiltonian and NTK spectra. 

\subsection*{V.B The Sign Bias}

As was mentioned above, representing the ground state of a Hamiltonian
using a neural network is a process that is basis dependent. It can
be seen by examining Eq. (\ref{eq:8}) and noticing that it is not
invariant under a change of basis. This is caused by the presence
of the non-trivial NTK in the equation, which is the same regardless
of the basis we choose to represent to Hamiltonian in. One can also
see that the mean initial values such as the energy and the ground
state overlap are also not invariant under a change of basis due to
the presence of the CK. So although there is no physical meaning to
applying a basis transformation to our Hamiltonian, the method is
sensitive to the chosen basis, As a result, we will now try to provide
some understanding of which basis are preferred by presenting the
following results, with derivations given in \hyperref[subsec:App.C]{App. C}. 

We first show which basis transformations would not effect the performance
of the method
\begin{thm}
1. For every unitary $U$ such that $\left[U,\Theta\right]=0$, the
solution for the transformed Hamiltonian $\tilde{H}=UHU^{\dagger}$
is the transformed solution $\ket{\tilde{\psi}}=U\ket{\psi}$, with
$\ket{\psi}$ being the solution for the Hamiltonian $H$. \\
2. For every unitary $U$ such that $\left[U,K\right]=0$, the average
initial values of the energy and ground state overlap are invariant.
\end{thm}
The theorem above means that basis transformations that commute with
both kernels will yield the same overall performance. By using the
known kernels properties in Sec. \hyperref[sec:III]{III} it immediately
follows that 
\begin{cor}
Given two Hamiltonians $H$ and $\tilde{H}$ related by a unitary
transformation $\tilde{H}=UHU^{\dagger}$ of the form $U=e^{i\chi}$
with $\chi=\sum_{S}\alpha_{S}X_{S}$, for any real $\alpha_{S}$ and
$S\subset\left\{ 1,...,N\right\} $, The solutions are related by
$\ket{\psi}=U\ket{\tilde{\psi}}$ with the same average initial conditions.
Here $\ket{\psi},\ket{\tilde{\psi}}$ are the solutions for Eq. (\ref{eq:8})
under the Hamiltonians $H,\tilde{H}$ respectively.
\end{cor}
A specific class of unitary transformation of the form defined above
are just $X_{S}$ with $S$ being any subset of qubits. Their effect
is to change the signs of operators involving the $Z$ and $Y$ Pauli
matrices in the Hamiltonian, since $-Z=XZX$.

In order to demonstrate the effect of different basis we will focus
on Hamiltonians of the form 
\begin{equation}
\mathcal{H}=\sum_{S}\alpha_{S}X_{S}+\beta_{S}Z_{S}+\gamma_{S}Y_{S},\label{eq:22}
\end{equation}
with $\alpha_{S},\beta_{S},\gamma_{S}$ being real, $S$ being any
subset of qubits, and $\gamma_{S}$ being non zero for subsets with
even cardinality, in order to make sure the Hamiltonian is real. Although
the family is not the most general one can consider, it includes many
models of interests including the generalized TFIM, Heisenberg, or
the J1-J2 model defined below in Eq. (\ref{eq:33}). 
\begin{cor}
For the family of Hamiltonians of the form $\mathcal{H}$, the average
solution is invariant under a unitary change of the signs of $\beta_{S}$
and/or $\gamma_{S}$ while keeping $\alpha_{S}$ constant. 
\end{cor}
The above follows immediately from the fact that altering the signs
of $\beta_{S}$ and $\gamma_{S}$ using a unitary matrix without altering
$\alpha_{S}$ can only be done by applying a unitary of the form $X_{S}$.
One consequence of this result is that changing the signs over the
diagonal using unitary transformation will not affect the performance.

On the other hand, changing the signs of $\alpha_{S}$ might change
the performance based on the results presented in previous sections.
Changing the signs of $\alpha_{S}$ is usually achieved in a unitary
fashion using ``sign transformations'' of the form $Z_{s}$, since
$ZXZ=-X$. 
\begin{prop}
Whenever $\Theta$ and $K$ are entry-wise positive, for the family
of Hamiltonians $\mathcal{H}$ in Eq. (\ref{eq:22}), the following
holds: \\
(1) If there is a sign transformation of the form $Z_{S}$ that will
make $\alpha_{S}\leq0\ \forall S$, then it will minimize $\mathcal{N}\left(\mathcal{H},\Theta\right)$
and $\mathcal{N}\left(\mathcal{H},K\right)$ out of all possible $Z_{S}$
transformations.\\
(2) For $\lvert\gamma_{S}\lvert\leq\lvert\alpha_{S}\lvert$, if there
is a sign transformation of the form $Z_{S}$ that will make $\alpha_{S}\leq0\ \forall S$,
then it will maximize $\mathcal{M}\left(\mathcal{H},\Theta\right)$
and $\mathcal{M}\left(\mathcal{H},K\right)$ out of all possible $Z_{S}$
transformations.\\
(3) For $\lvert\gamma_{S}\lvert\leq\lvert\alpha_{S}\lvert$, if there
is a sign transformation of the form $Z_{S}$ that will make $\alpha_{S}\leq0\ \forall S$,
then it will maximize the learning rate bound from Sec. \hyperref[subsec:IV.D]{IV.D}
in the large bias approximation as defined in Sec. \hyperref[subsec:IV.D]{IV.D},
out of all possible $Z_{S}$ transformations.
\end{prop}
Another very common family of Hamiltonians is that of 2-local Hamiltonians,
which can be written as 
\begin{equation}
H_{2}=\sum_{\gamma}h_{i}^{(\gamma)}\sigma_{i}^{(\gamma)}+\sum_{\gamma\delta}\sum_{i,j}J_{i,j}^{(\gamma,\delta)}\sigma_{i}^{(\gamma)}\sigma_{j}^{(\delta)},\label{eq:23}
\end{equation}
with $\sigma_{i}^{(\gamma)}\in\left\{ X_{i},Y_{i},Z_{i}\right\} $
and $i\in\left\{ 1,...,N\right\} $. One can also show that 
\begin{prop}
For two-local Hamiltonians, if there is a sign transformation of the
form $Z_{S}$ that will make $J_{i,j}^{(X,X)},h_{i}^{(X)}\leq0\ \forall i,j$,
then it will minimize $\mathcal{N}\left(H_{2},\Theta\right)$ and
$\mathcal{N}\left(H_{2},K\right)$ out of all possible $Z_{S}$ transformations.
\end{prop}
For general Hamiltonians one can also show the following
\begin{prop}
For any Hamiltonian $H$, if $\Theta$ and $K$ are entry-wise positive,
$\braket{g|\Theta|g}$ and $\braket{g|K|g}$ are maximized when $\braket{\sigma|g}\geq0\ ,\forall\sigma$
out of all possible signs of $\braket{\sigma|g}$ with the same norms.
\end{prop}
As mentioned in Sec. \hyperref[sec:III]{III}, $K$ is entry-wise
positive whenever the activation functions are positive, while $\Theta$
is entry-wise positive whenever we have a high bias toward $\Theta_{0}$
in the NTK spectra, both of which are very common cases. 

The results above provide an explanation for why applying a transformation
to make the Hamiltonian obey the Marshall-Peierls sign rule before
optimizing the network enhanced the performance in previous studies
\citep{Park_2022,Nomura_2021,Bukov_2021}. Such a transformation is
a specific case of using a $Z_{S}$ transformation in order to turn
the Hamiltonian into a stoquastic Hamiltonian. Thus the above demonstrates
that for the general family of Hamiltonians defined above, if one
can make the Hamiltonian stoquastic using a $Z_{S}$ transformation,
then it will tend to improve the performance of the optimization,
making the stoquastic form preferred whenever it is attainable. 

More generally, it was shown empirically that stoquastic Hamiltonians
are easier for NQS ground state optimizations \citep{Bukov_2021,Szabo:2020vk},
as compared to highly non-stoquastic Hamiltonians. One of the reasons
for this phenomenon is the fact that for a lot of common ANN architectures,
the kernels will be highly biased toward their principal component.
From the kernels property in Eq. (\ref{eq:14}), the principal component
always have a simple sign structure, and in most common cases, it
is the state $\ket{+}=\frac{1}{2^{N/2}}\sum_{\sigma}\ket{\sigma}$
which is ``in the middle'' of the positive cone. Since stoquastic
Hamiltonians always have a positive ground state, $\braket{g|\sigma}\geq0$,
the bias is compatible with their ground states, in the sense that
the kernel expectation values over the ground state will be large
compare to non-stoquastic Hamiltonians. One can see that by considering
the extremely biased case toward the $\ket{+}$ state, where $\mathcal{M}\left(H,\Theta\right)\propto\braket{g|\Theta|g}\approx\left(\sum_{\sigma}\braket{\sigma|g}\right)^{2}.$
Thus, in general ground state with a uniform sign structure will have
larger kernel expectation then a ground state with a highly non-uniform
signs.

On top of that, stoquastic Hamiltonians will yield low values for
$\mathcal{N}\left(H,\Theta\right)$ and $\mathcal{N}\left(H,K\right)$,
meaning that they will be strongly compatible with the kernels, in
the sense that lower energy states will have higher kernels expectation
values. One can see that for example by considering the most general
form of two-local Hamiltonians in Eq. (\ref{eq:23}), which are the
most common local Hamiltonians in practical settings. For such Hamiltonians 

\[
\text{\ensuremath{\mathcal{N}\left(H_{2},\Theta\right)}=\ensuremath{\sum_{i,j}J_{i,j}^{(X,X)}\alpha_{2}}}+h_{i}^{(X)}\alpha_{1}
\]
where we assumed that $\text{Tr}\left(H\right)=0,\Vert H\Vert_{F}=\Vert\Theta\Vert_{F}$
W.L.G. $\alpha_{1},\alpha_{2}$ depend on $\Theta$, and whenever
there is a significant bias toward $\ket{+}$ they are both positive.
A necessary condition for a 2-local Hamiltonians to be stoquastic
is that $J_{i,j}^{(X,X)},J_{i,j}^{(X,I)},J_{i,j}^{(X,J)}$ will all
be negative. And a negative value of those coefficients will in general
produce a low value for $\mathcal{N}\left(H,\Theta\right)$ compared
to a highly non-stoquastic Hamiltonians with highly non-uniform signs
for those coefficients. 

Even in the most general case, without assuming anything about the
kernels, using the weak bias property in Eq. (\ref{eq:14}), there
will always be some bias toward states with less ``complex'' sign
structures. As mentioned, the kernels eigenstates are the states $\ket{s}$
with $s$ being a string of $+$ and $-$. The kernels spectra decomposes
into sectors, each one of which is a subspace spanned by the states
$\ket{s}$ with a specific number of $+$ signs. Each sector has a
unique complexity as was defined in Sec. \hyperref[sec:III]{III}
and \citep{https://doi.org/10.48550/arxiv.1907.10599}, which is the
number of Hilbert space dimensions that the state oscillates in. The
weak bias property of every NTK and CK states that there will always
be some bias toward a ``simple'' subspace with complexity 0 or 1.
Thus, ground states which have their most significant component in
low complexity sectors will be more aligned with the bias, while highly
``oscillating'' ground states will be much less aligned, rendering
them harder to learn. Due to the exponential growth of Hilbert space,
the bias toward the principal component will grow exponentially with
the system size.

\subsection*{V.C Basis Optimization \protect\label{subsec:V.C}}

One can utilize the relations between the metrics $\mathcal{N}\left(H,\Theta\right)$
and $\mathcal{N}\left(H,K\right)$ and the overall performance in
order to find the local basis most suitable for the ANN at hand. In
the case of 2-local Hamiltonians, every local basis transformation,
described by a set of unitary matrices $U_{i}\in SU(2)\ ,i\in\left\{ 1,...,N\right\} $,
such that $\tilde{H}=UHU^{\dagger}$ with $U=\otimes_{i}^{N}U_{i}$,
corresponds to a set of orthogonal matrices $O_{i}\in O(3)$, that
transform the Hamiltonian coefficients as \citep{Klassen_2019}
\begin{equation}
\tilde{h}_{i}^{(\gamma)}=\sum_{\delta}O_{i}^{(\gamma,\delta)}h_{i}^{(\delta)}\ ,\ \tilde{J}_{i,j}^{(\gamma,\delta)}=\sum_{\alpha,\beta}O_{i}^{(\gamma,\alpha)}J_{i,j}^{(\alpha,\beta)}O_{j}^{(\beta,\delta)}.
\end{equation}
For example, one can then try to optimize $\mathcal{N}\left(\tilde{H},K\right)$
over these orthogonal matrices. Since under a basis transformation
the trace is preserved, ones has
\begin{equation}
\mathcal{N}\left(\tilde{H},K\right)=\text{Tr}\left(\tilde{H}K\right)=k_{2}\sum_{i,\delta}O_{i}^{(X,\delta)}h_{i}^{(\delta)}+k_{1}\sum_{i,j,\alpha,\beta}O_{j}^{(X,\alpha)}J_{i,j}^{(\alpha,\beta)}O_{j}^{(\beta,X)},
\end{equation}
with $k_{1}\equiv\text{Tr}\left(X_{1}K\right)$ and $k_{2}\equiv\text{Tr}\left(X_{1}X_{2}K\right)$,
which depends on the network. This optimization is not easy in general.
In the translationally invariant case one can restrict oneself to
a uniform transformation, meaning that $O_{i}\equiv O$, which reduces
the number of variables to a constant. Consider for example the TFIM
in Eq. (\ref{eq:13}), and consider transformations which will render
the transformed Hamiltonian real for the sake of simplicity. Then
the transformation takes place in the $XZ$ plane, meaning that $O\in O(2)$.
One can further verify that in our case it is enough to consider $O\in SO(2)$,
which is parameterized by a single parameter $\theta$. Carrying out
the calculation one obtains 
\begin{equation}
\mathcal{N}\left(\tilde{H},K\right)=-k_{2}J\sin^{2}(\theta)-k_{1}hN\cos(\theta),
\end{equation}
with $J$ being the number of edges over the interactions graph. Optimizing
this expression in term of $\theta$, one finds that the optimal value
of $\mathcal{N}\left(\tilde{H},K\right)$ is $\min\left\{ -k_{2}J-\frac{k2N^{2}h^{2}}{4K_{2}J},\ -k_{1}Nh,\ k_{1}Nh\right\} $,
which corresponds to $\theta=\arccos\left(-\frac{hk_{1}N}{2k_{2}J}\right),\ \theta=0,\ \theta=\frac{\pi}{2}$,
respectively. $k_{1}$ and $k_{2}$ can be approximated in the extremely
biased regime when the bias is toward $\ket{+}$ such that $k_{1}\approx k_{2}$.
In other cases, for common activation functions and architectures
one can use various method to evaluate them approximately or extrapolate
from small system sizes \citep{https://doi.org/10.48550/arxiv.1711.00165,neal2012bayesian,https://doi.org/10.48550/arxiv.1907.10599,https://doi.org/10.48550/arxiv.2104.03093}.

\section*{VI. Numerical Results \protect\label{sec:VI}}

We demonstrate the ideas presented above in numerical experiments
in the infinite limit. Specifically, we will show that the metrics
$\mathcal{M}\left(H,\Theta\right)$, $\mathcal{N}\left(H,\Theta\right)$,
$\mathcal{M}\left(H,K\right)$, and $\mathcal{N}\left(H,K\right)$
are indeed correlated, respectively, with the convergence time of
the network to the ground state, their ground state energy, and the
initial values of both these quantities.

\subsection*{VI.A NTK vs. Convergence Time}

$\mathcal{N}\left(H,\Theta\right)$ is a version of the inner product
between the Hamiltonian and NTK, and thus provides an a-priori metric
about the overall relationship between the Hamiltonian and NTK spectra
and eigenstates. As conjectured, the lower it is, the faster the energy
will decrease overall, due to the bias toward lower energy eigenstates
in general. 

$\mathcal{M}\left(H,\Theta\right)$ is proportional to the expectation
value of the NTK over the ground state, so it provides a measure for
the bias of the optimization toward to ground state, as compared to
the principal component. It is important to note that $\mathcal{N}\left(H,\Theta\right)$
can be small, indicating that the bias will be compatible with lowering
the overall energy, while $\mathcal{M}\left(H,\Theta\right)$ is also
small, so obtaining a low energy state will be easy, yet obtaining
the actual ground state will be hard. As noted before, these metrics
are only sensible when one uses them to compare different bases, while
keeping the initial conditions the same. 

In this section we will consider different models and different architectures
under a local basis transformation, which is determined by a single
parameter in a similar fashion to Sec. \hyperref[subsec:V.C]{V.C}.
In order to eliminate the influence of the initial conditions determined
on average by the CK, we use the initial uniform superposition 
\begin{equation}
\ket{\psi_{0}}=\sum_{i}\ket{E_{i}},
\end{equation}
with $\ket{E_{i}}$ being the eigenstates of the Hamiltonian. 

In Fig. \ref{fig:Left:-The-energy} we consider the noninteracting
Hamiltonian 
\begin{equation}
H_{\text{{local}}}=-\sum_{i}X_{i},
\end{equation}
which under local basis rotation assumes the form 
\begin{equation}
H_{\text{{local}}}=\sum-\cos\left(\theta\right)X-\sin\left(\theta\right)Z.\label{eq:=00002029}
\end{equation}
Following Sec. \hyperref[subsec:V.C]{V.C} we expect that the smaller
$\theta$ is, the faster the NQS will converge to the ground state;
As can be seen from the results, this is indeed the case. We also
see that the metric $\mathcal{M}\left(H,\Theta\right)$ is indeed
positively correlated with the convergence time. 

\begin{figure}
\begin{centering}
\includegraphics[scale=0.22]{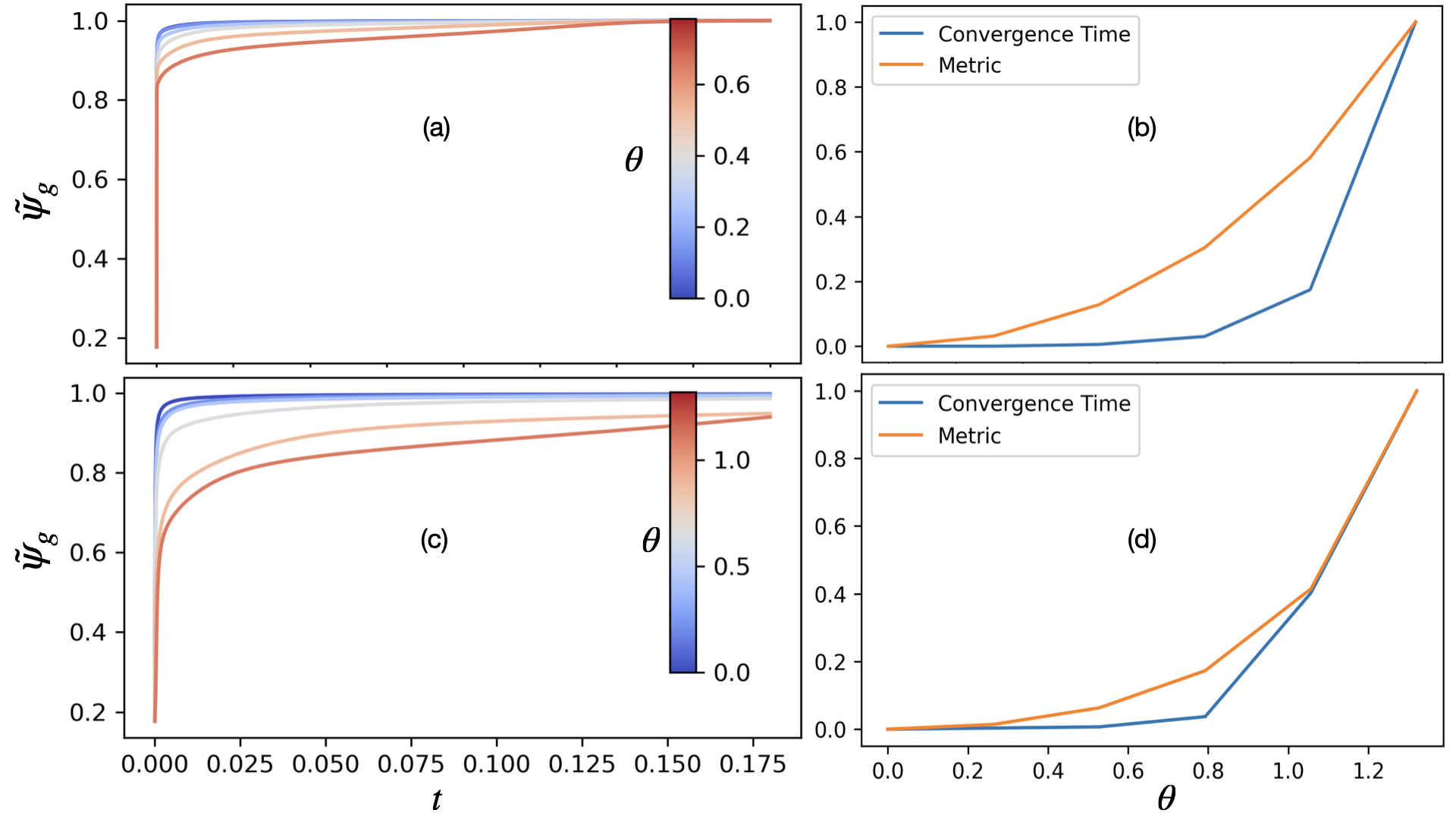}\caption{\protect\label{fig:Left:-The-energy}(a) The ground state overlap
as a function of time for $H_{\text{TFIM}}$ with varying $\theta$.
(b) The convergence time, which is the time it takes for the overlap
to reach a fixed threshold $\epsilon=0.95$ and the metric $\mathcal{M}\left(H,\Theta\right)$
as a function of $\theta$. (c) The ground state overlap as a function
of time for $H_{\text{local}}$ for $\theta\in\left[0,\frac{\pi}{2}\right]$.
(d) The convergence time and the metric $\mathcal{M}\left(H,\Theta\right)$
plotted as a function of $\theta$. Both networks have one hidden
layer and a ReLU activation. Metrics and convergence times are normalized
to the range $\left[0,1\right]$.}
\par\end{centering}
\end{figure}
We use the same metric in Fig. \ref{fig:Left:-The-energy} for the
TFIM model in Eq. (\ref{eq:13}), with $h_{i}=2$. Under local basis
rotation, from the analysis in Sec. \hyperref[sec:IV.C]{IV.C} we
expect the fastest convergence to occur when $\theta=0$, which is
indeed the case. This can be understood from the more general analysis,
since this is the case where the negative off diagonal elements are
most dominant.

We also examine if the metric $\mathcal{N}\left(H,\Theta\right)$
is positively correlated with the convergence time of the energy to
the ground state energy. We first examine this relationship for the
Heisenberg model with a transverse field,
\begin{equation}
H_{\text{Heis}}=\sum_{\braket{i,j}}\vec{\sigma}_{i}\cdot\vec{\sigma}_{j}-h\sum_{i}X_{i},
\end{equation}
with $h=1$. We note that the transverse field is introduced since
otherwise the Hamiltonian would be invariant under uniform local basis
transformations. Due to the symmetry of the coupling term we expect
to get similar results to the uncoupled chain, such that the minima
will occur at $\theta=0$. As can be seen in Fig. \ref{fig:Left:-The-average-1},
this is indeed the case. One can also see that the metric $\mathcal{N}\left(H,\Theta\right)$
is indeed positively correlated with the convergence time. We further
demonstrate this relation for the XYZ Hamiltonian 

\begin{figure}
\begin{centering}
\includegraphics[scale=0.22]{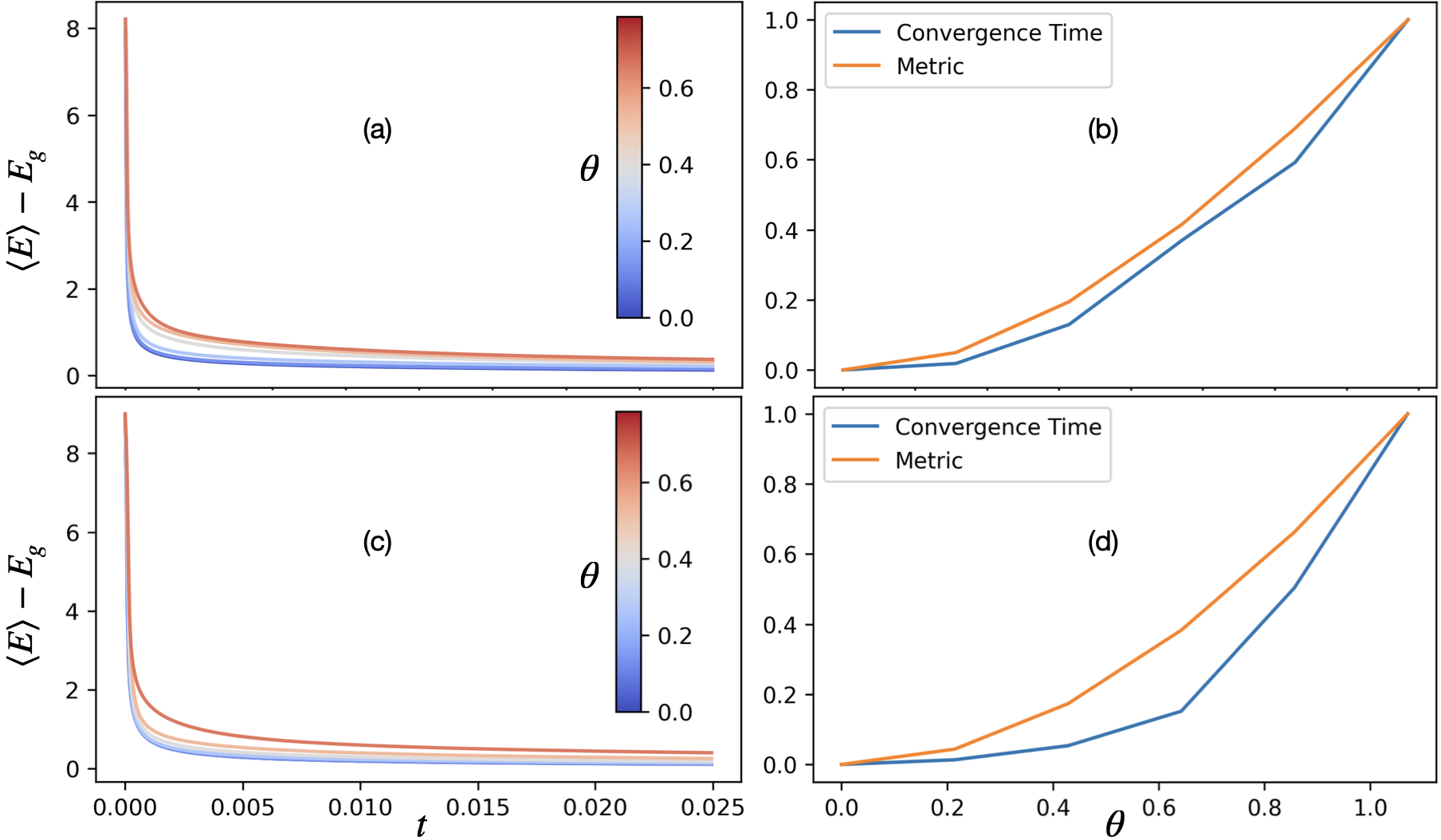}\caption{\protect\label{fig:Left:-The-average-1}(a) The energy as a function
of time for $H_{\text{XYZ}}$ for different local basis. (b) The average
convergence time and the metric $\mathcal{N}\left(H,\Theta\right)$
as a function of $\theta$ normalized to the range $\left[0,1\right]$.
(c) The energy as a function of time for $H_{\text{Heis}}$ for different
local basis. (d) the average convergence time and the $\mathcal{N}\left(H,\Theta\right)$
as a function of $\theta$ normalized to the range $\left[0,1\right]$.
Both networks consist of one hidden layer with ReLU activation. Metrics
and convergence times are normalized to the range $\left[0,1\right]$.}
\par\end{centering}
\end{figure}
\begin{equation}
H_{\text{XYZ}}=\sum_{\braket{i,j}}\alpha X_{i}X_{j}+\beta Z_{i}Z_{j}+\gamma Y_{i}Y_{j},
\end{equation}
with $\alpha=-2,\beta=-1,\gamma=-0.5$. Again one can observe in Fig.
\ref{fig:Left:-The-average-1} that the metric $\mathcal{N}\left(H,\Theta\right)$
is highly correlated with the convergence time for the energy expectation
over the NQS. 

In Fig. \ref{fig:Left:-the-convergence} we provide the results of
additional benchmarks with different models and architectures, which
all seem to indicate a strong relationship between the metrics and
the corresponding convergence time. 

\begin{figure}
\begin{centering}
\includegraphics[scale=0.2]{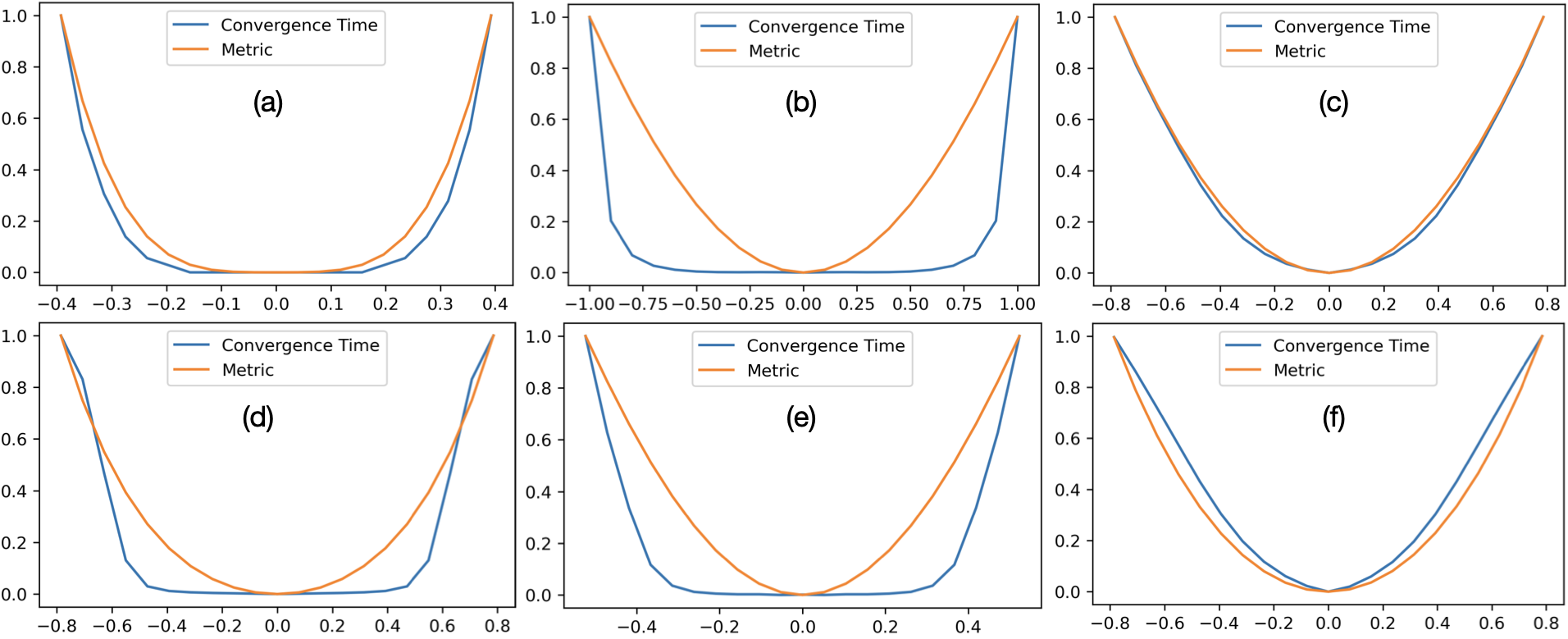}\caption{\protect\label{fig:Left:-the-convergence}(a) Ground state overlap
convergence time and $\mathcal{M}\left(H,\Theta\right)$ for $H_{\text{local}}$
with one hidden layer and tanh activation. Energy convergence time
and the metric $\mathcal{N}\left(H,\Theta\right)$ for $H_{\text{Heis}}$
with (b) one hidden layer and erf activation, and (c) two hidden layers
with ReLU activation. (d) Ground state overlap convergence time and
$\mathcal{M}\left(H,\Theta\right)$ for $H_{\text{TFIM}}$ with one
hidden layer and erf activation. (e) Energy convergence time and the
metric $\mathcal{N}\left(H,\Theta\right)$ for $H_{\text{XYZ}}$ with
one hidden layer and sigmoid activation function. (f) Energy convergence
time and the metric $\mathcal{N}\left(H,\Theta\right)$ for $H_{\text{TFIM}}$
with two hidden layers and ReLU activation function. Metrics and convergence
times are normalized to the range $\left[0,1\right]$}
\par\end{centering}
\end{figure}
Beyond the general correlation between the metrics and the convergence
time, one can see that in all of the models we examined, local minima
for the convergence time are aligned with the local minima for the
metrics. All of those local minima occur when the negative off diagonal
terms are the most dominant relative to other bases. Further details
are provided in \hyperref[subsec:App.D]{App. D}.

\subsection*{VI.B CK vs. Initial Averages }

As was mentioned in Sec. \hyperref[subsec:III.B]{III.B}, the CK is
just the covariance matrix of the gaussian distribution of the initial
unnormalized states. The metrics $\mathcal{N}\left(H,K\right)$ and
$\mathcal{M}\left(H,K\right)$ are just the large-$N$ aproximations
to the average initial energy and ground state overlap, respectively. 

In Fig. \ref{fig:Top:-Average-initial} we plot the averages and their
corresponding metrics in several cases. Similarly to the previous
subsection, we use the uncoupled Hamiltonian as defined in Eq. (\ref{eq:=00002029})
under local unitary basis transformation. We also examine the dependence
over a model parameter such as the strength of the transverse field
in the case of the TFIM. We observe that the more the negative off-diagonal
elements are dominant, the better the metrics and averages. We also
note that there is a crucial difference between negative and positive
off diagonal elements in the initial ground state overlap. Finally
we use the 1D $J1\text{-}J2$ model,

\begin{equation}
H_{J1\text{-}J2}\equiv J_{1}\sum_{\braket{i,j}}\vec{\sigma}_{i}\vec{\sigma}_{j}+J_{2}\sum_{\braket{\braket{i,j}}}\vec{\sigma}_{i}\vec{\sigma}_{j},\label{eq:33}
\end{equation}
with $\braket{\braket{i,j}}$ indicating next nearest neighbors, with
$J_{1}=1$ and $J_{2}=\frac{1}{2}$. One can see that the ground state
overlap prediction is quite accurate, and that it displays a rapid
decrease as a function of $N$. Although we use a 1D example we can
already see how the metrics can predict the hardness of trying to
simulate the $J1\text{-}J2$ model using ANN. 

\begin{figure}

\begin{centering}
\includegraphics[scale=0.18]{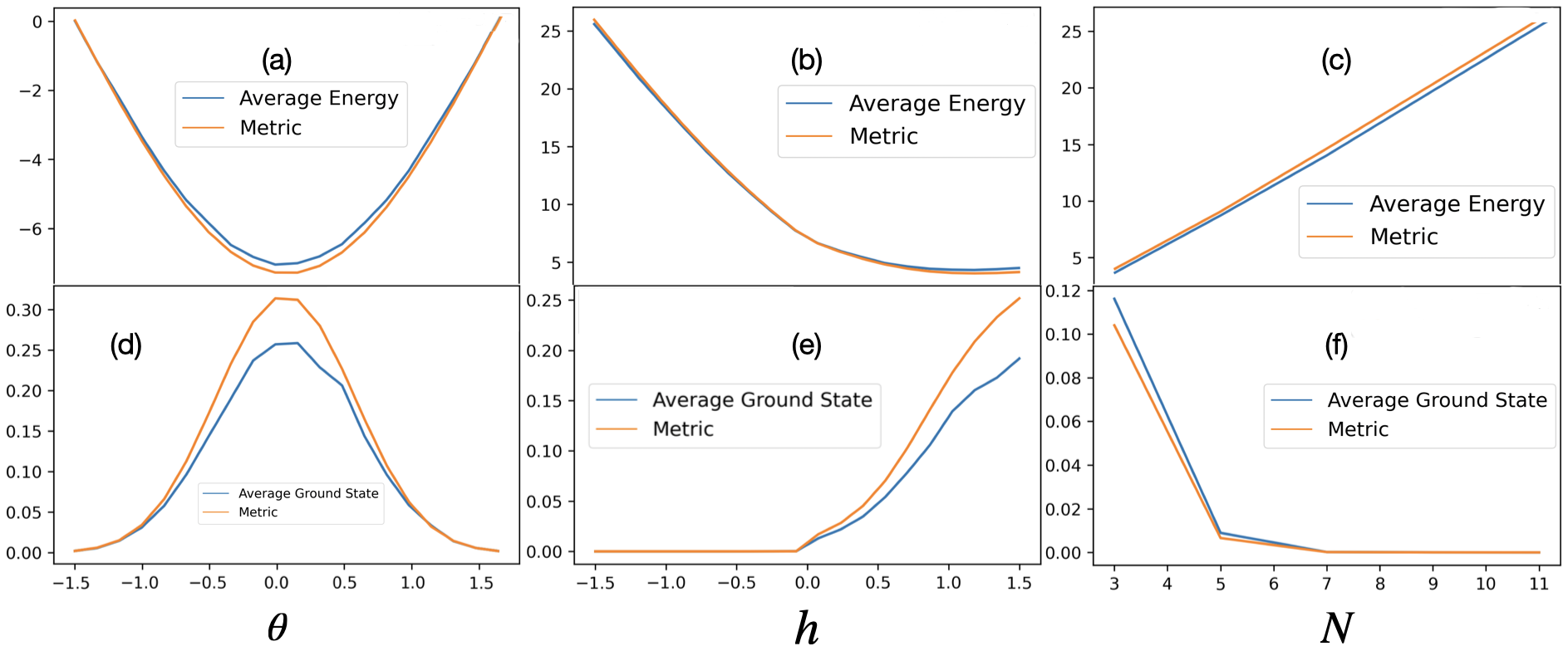}\caption{\protect\label{fig:Top:-Average-initial}(a,b,c) Average initial energy
and (d,e,f) ground state overlap together with their corresponding
metrics $\mathcal{N}\left(H,K\right)$ and $\mathcal{M}\left(H,K\right)$,
averaged over $10^{3}$ samples. (a,d) The model $H_{\theta}$ with
$\theta\in\left[0,\pi\right]$, where we omitted the $E_{g}$ term
from $\mathcal{N}\left(H,K\right)$ due to it being independent of
$\theta$. (b,e) The TFIM as a function of $h\in\left[-1.5,1.5\right]$,
where negative values of $h$ correspond to positive off diagonal
elements in the Hamiltonian and vice versa.  (c,f) The $J1\text{-}J2$
model at with $J1/J2=0.5$ as a function of $N$. For all of the above
we used one hidden layer and ReLU activation function. }
\par\end{centering}
\end{figure}
In Fig. \ref{fig:Top:-Average-initial-1} we add further benchmarks.
One can see that in some cases there can be large deviations between
the metric and the actual average, due to the rough estimation the
metrics were derived from. Nonetheless, their trends perfectly match.
This can be used for the purpose of comparing different parameters,
bases, or architectures.

\begin{figure}
\begin{centering}
\includegraphics[scale=0.19]{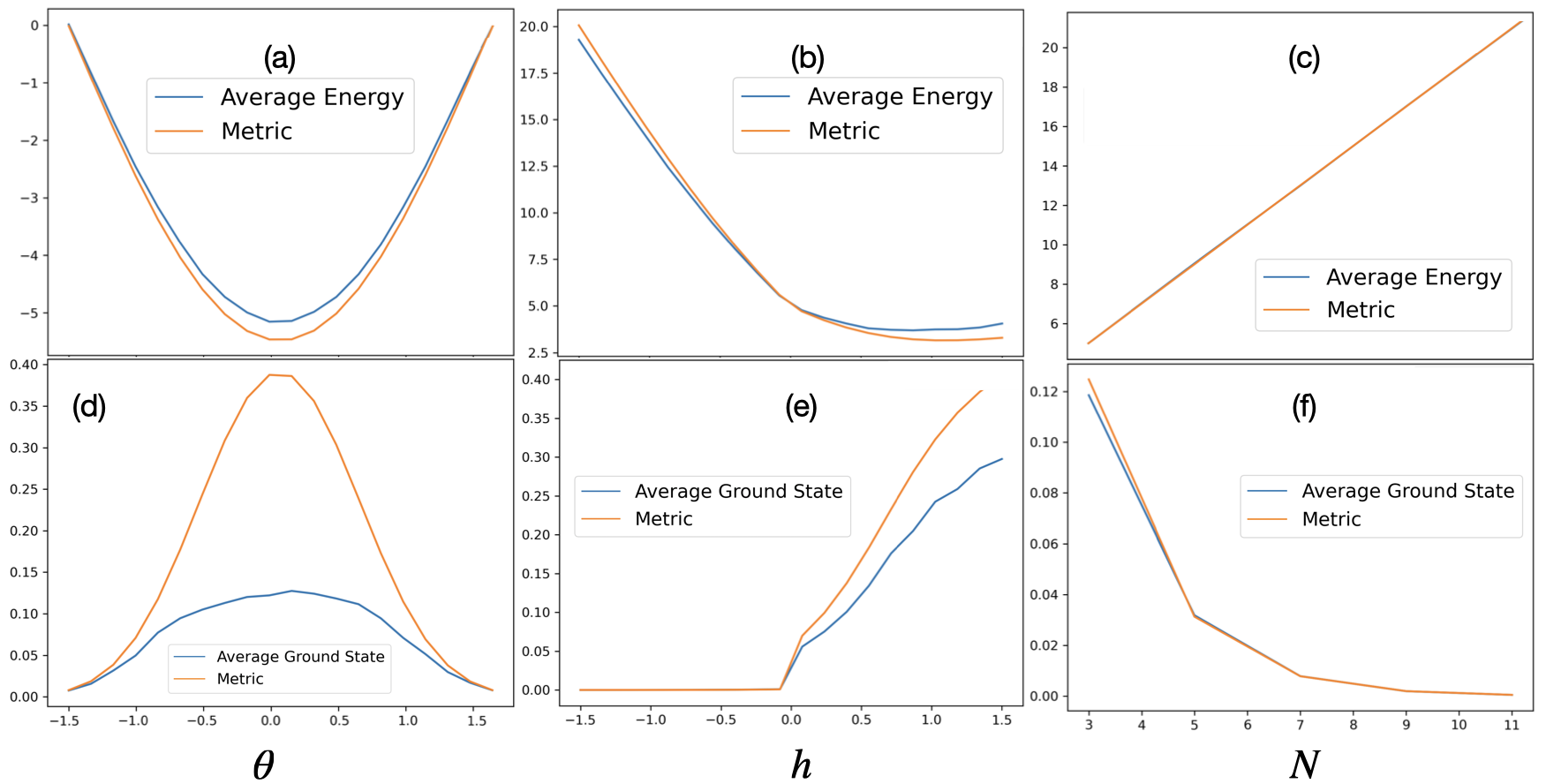}
\par\end{centering}
\caption{\protect\label{fig:Top:-Average-initial-1}(a,b,c) Average initial
energy and (d,e,f) ground state overlap with their corresponding metrics.
(a,d) $H_{\theta}$ with $\theta\in\left[0,\pi\right]$, with tanh
activation function and one hidden layer. (b,e) $H_{\text{TFIM}}$
with $h\in\left[-1.5,1.5\right]$, with two layers and ReLU activation.
(c,f): $H_{\text{Ising}}$ with sigmoid activation function. }

\end{figure}

\section*{VII. Stochastic Reconfiguration \protect\label{sec:VII}}

The non-triviality of the NTK arrises from the fact that the gradient
flow is performed in parameter space. A method known more generally
as natural gradient and in our specific context as Stochastic Reconfiguration
(SR) \citep{becca_sorella_2017,Sorella:1998vk,Park_2020} attempts
to mimic a gradient flow in the wavefunction space (which in our case
is just the Hilbert space). It is usually achieved by a modification
of the parameter update such that 
\begin{equation}
\frac{d\theta_{i}}{dt}=\sum_{j}\left(J^{T}J+\epsilon I\right)_{ij}^{-1}\frac{d\mathcal{L}}{d\theta_{j}},
\end{equation}
with $J_{\sigma i}=\frac{\partial\psi_{\sigma}}{\partial\theta_{i}}$.
Can it alleviate the kernel bias discussed above? Using the chain
rule we get 
\begin{equation}
\frac{d\ket{\psi}}{dt}=\frac{1}{\braket{\psi|\psi}}\left(I+\epsilon\Theta^{-1}\right)^{-1}\left(\braket{E}-H\right)\ket{\psi},
\end{equation}
which for a sufficiently small $\epsilon$, in the limit of infinite
width, is indeed the imaginary time Schrödinger equation. Although
it might seem as a remedy to the kernel bias, it is important to notice
several points. The first is that the initialization is still governed
by the CK, and biased in the same manner, so using the usual initialization
strategies will produce a poor initial overlap with the ground state
when the CK expectation value over the ground state is low. 

Another thing that one can notice from the equation above, is that
it is the same as Eq. (\ref{eq:8}), with $\Theta_{\text{eff}}=\left(I+\epsilon\Theta^{-1}\right)^{-1}\approx I-\epsilon\Theta^{-1}$.
One can then see that in order to guarantee that the resulting equation
will mimic the imaginary time Schrödinger equation one needs to have
\begin{equation}
\epsilon\ll\Theta_{\min}
\end{equation}
with $\Theta_{\min}$ being the minimal eigenvalue of $\Theta$. We
note that using the results listed in Sec. \hyperref[sec:III]{III},
$\Theta_{\min}$ decays as $\sim\frac{1}{2^{N^{2}}}$, so one will
need a very small $\epsilon$ to fully mimic the imaginary time dynamics.
If $\epsilon$ is not small enough, then the effective NTK affect
can be non-trivial and might induce a significant bias. Following
the analysis presented before, in order for the dynamical bias to
be aligned with the ground state, $\braket{g|\Theta_{eff}|g}\approx1-\epsilon\braket{g|\Theta^{-1}|g}$
should come close to unity as much as possible, so $\braket{g|\Theta^{-1}|g}$
should be small as possible, similar to the case without the SR method.

It is also important to note that these observations should also be
weighted against the fact that since the application of the SR method
requires a matrix inversion, each parameter update requires $\mathcal{O}\left(p^{3}\right)$
operations with $p$ being the number of parameters, so in general
it is mostly applicable to small systems \citep{Park_2020}.

\section*{VIII. General Ansätze \protect\label{sec:VIII}}

In this section, we discuss how our approach for the study of the
possible existence of a sign bias, or, more generaly, the basis dependence,
might relate to other ansatz that are optimized under gradient descent. 

For a general ansätze, the dyanmics would be governed by the equation
\begin{equation}
\frac{d\ket{\psi}}{dt}=\frac{1}{\braket{\psi|\psi}}\Theta_{t}\left(\braket{E}-H\right)\ket{\psi},
\end{equation}
with the NTK being potentially time dependent now, and where the ansätze
might not be able to represent every state in Hilbert space. 

We begin by asking what are the criteria for the performence of the
ansätze to be invariant under some change of basis, represented by
a unitary transformation $U$. We denote the solution the equation
above with $H$ by $\psi_{\sigma}\left(\theta(t)\right)$ and the
solution for the transformed Hamiltonian $\tilde{H}=UHU^{\dagger}$
by $\psi_{\sigma}\left(\tilde{\theta}(t)\right)$, with $\theta$
and $\tilde{\theta}$ being the parameters of the two solutions. 
\begin{enumerate}
\item The initial average (over the distribution of the function parameters)
of the every function of interest, e.g., the energy expectation value
or the ground state overlap should by invariant under U. For a general
Hamiltonian this would require 
\[
p_{0}\left(\psi_{\sigma}\right)=p_{0}\left(\sum_{\sigma'}U_{\sigma\sigma'}\psi_{\sigma'}\right)
\]
with $p_{0}$ being the initial distribution of the wavefunction.
\item For all $t$ there is a set of parameters $\tilde{\theta}$ such that
$\psi_{\sigma}\left(\tilde{\theta}(t)\right)=\sum_{\sigma'}U_{\sigma\sigma'}\psi_{\sigma}\left(\theta(t)\right)$.
\end{enumerate}
In general, verifying that these criteria are obeyed is not easy.
For infinite NQS, the second criteria is trivial, and the first criterion
reduces to the much simpler criterion that the NTK and CK will commute
with $U$ as stated on Sec. \hyperref[sec:III]{III} due to the emergence
of gaussian distribution for the initial state and the convergence
of the NTK to its mean initial value. In the following subsection
we show a very well known ansätze that can be seen to satisfy those
criteria for every local unitary transformation 

\subsection*{VIII.A Matrix Product States are Local-Basis Unbiased}

Let us consider the most elementary tensor network, the Matrix Product
States (MPS), which is defined algebraically as
\begin{equation}
\psi_{\sigma}=\sum_{i_{1,}...i_{N}}\left(\prod_{n}^{N}A_{i_{n},i_{n+1}}^{(\sigma_{n})}\right),\label{eq:38}
\end{equation}
for a system of $N$ spins, with each $A_{i_{n},i_{n+1}}^{(\sigma_{n})}$
being an order 3 tensor. The upper index $\sigma_{n}\in\left\{ -1,1\right\} $
corresponds to one of the physical spins, while the indices $i_{n}\in\left\{ 1,...,\chi\right\} $
are virtual (dummy) indices to be summed over. For simplicity we will
consider periodic boundary conditions where $i_{1}=i_{N}$. MPSs are
usually represented graphically using vertices and edges as follows 

\begin{figure}[H]
\centering

\centering{}\includegraphics[scale=0.3]{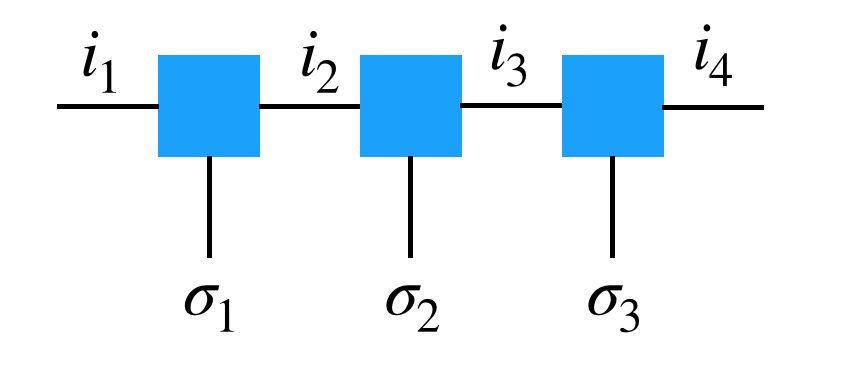}
\end{figure}
\hspace{0pt}\\
with each blue square with 3 black legs corresponding to a tensor
$A_{i_{n},i_{n+1}}^{(\sigma_{n})}$. The capacity of the ansatz is
controlled mainly by the geometry of the network and the virtual bond
dimension $\chi$. 

Lets us begin by showing the first criterion is met for local unitaries.
Usually the matrices $A^{(\sigma_{n})}$ are drawn i.i.d, from distributions
such as the Haar measure, matrix normal distributions and entrywise
normal distributions. If that is the case then 
\[
P(\psi_{\sigma})=\prod_{n}P\left(A^{(\sigma_{n})}\right)
\]
due to the i.i.d property. Thus, for local unitaries the first criterion
amounts to
\[
p\left(A^{(\sigma_{n})}\right)=p\left(\sum_{\sigma_{n}'}U_{\sigma_{n}\sigma_{n}'}A^{(\sigma_{n}')}\right).
\]
This criteria can easily be seen to be met whenever 
\[
p\left(A^{(\sigma_{n})}\right)=p\left(\sum_{\sigma_{n}}A^{(\sigma_{n})}A^{(\sigma_{n})^{T}}\right)
\]
,which is met for the Haar measure, matrix normal distribution with
identity covariance, entrywise normal distribution, and every other
distribution for the matrix distribution which is invariant under
the application of a unitary. 

The second criterion is immediate for every local unitary transformation.
Indeed, if we apply a local unitary $U$ on a spin $n$, the MPS will
transform into another MPS with the same bond dimension and local
matrix $\tilde{A}^{(\sigma_{n})}=\sum_{\sigma_{n}'}U_{\sigma_{n}\sigma_{n}'}A^{(\sigma_{n}')}$.
Hence, MPSs are locally unbiased and will not suffer from a sign bias
that can be solved by local change of basis. 

In the case that the MPS is overparameterized, i.e., that $\chi$
is large enough such that it can represent every $N$ spins wave function,
the first criterion is met for every unitary $U$, even a non-local
one, under mild conditions. Considering Eq. (\ref{eq:38}), since
there are a total of $N\chi$ product terms in the sum, the wave function
$\psi_{\sigma}$ will be distributed according to a normal distribution
by virtue of the central limit theorem. For a fixed $N$, the deviation
from the central limit theorem will scale with the bond dimension
as $\chi^{-N/2}$ in the case that each $A^{(\sigma_{n})}$ is entrywise
i.i.d, and at least as $\chi^{-1/2}$ if we allow correlations. The
covariance matrix for the normal distribution is analgous to the CK
and is given by 
\[
K=\left<\psi_{\sigma}\left(0\right)\psi_{\sigma'}(0)\right>.
\]
using Eq. (\ref{eq:38}) and the independence we obtain
\[
\left<\psi_{\sigma}\left(0\right)\psi_{\sigma'}(0)\right>=\left<\sum_{i_{1,}...i_{N}}\left(\prod_{n}^{N}A_{i_{n},i_{n+1}}^{(\sigma_{n})}\right)\sum_{i_{1,}...i_{N}}\left(\prod_{n}^{N}A_{i_{n},i_{n+1}}^{(\sigma'_{n})}\right)\right>
\]
\[
=\sum_{i_{1,}...i_{N},i'_{1,}...i'_{N}}\prod_{n}^{N}\left<A_{i_{n},i_{n+1}}^{(\sigma_{n})}A_{i'_{n},i'_{n+1}}^{(\sigma'_{n})}\right>=\sum_{i_{1,}...i_{N},i'_{1,}...i'_{N}}\prod_{n}^{N}\left<A_{i_{n},i_{n+1}}^{(\sigma_{n})}A_{i'_{n},i'_{n+1}}^{(\sigma'_{n})}\right>\delta_{\sigma_{n},\sigma'_{n}},
\]
where in the last equality we used the i.i.d property of the matrices
$A^{(\sigma_{n})}$. Finally we get 
\[
\left<\psi_{\sigma}\left(0\right)\psi_{\sigma'}(0)\right>=\left(\sum_{i_{1,}...i_{N},i'_{1,}...i'_{N}}\prod_{n}^{N}\left<A_{i_{n},i_{n+1}}^{(\sigma_{n})}A_{i'_{n},i'_{n+1}}^{(\sigma{}_{n})}\right>\right)\delta_{\sigma,\sigma'}\propto\delta_{\sigma,\sigma'}.
\]
Since the CK is proportional to the identity it will commute with
every unitary $U$. The second criterion is trivially obeyed in the
overparametrized limit.

\section*{IX. Conclusions and Future Outlook }

In this work we explored how the interplay between a neural network
architecture and a Hamiltonian affects the success of learning the
Hamiltonian ground state. Building upon recent developments in neural
network theory, we have used an infinite width limit which provides
twofold simplification; First, it provided mathematical simplicity
that allowed the study of the network initialization and dynamical
properties in a more tractable manner. Second, it absolved us from
going into the separate discussion of the network capacity. In this
spirit we also avoided going into the use of Monte Carlo techniques.
Although not practical, we could use these settings to reveal the
more fundamental obstacles that this method faces, which cannot be
alleviated by increased computational resources. 

Using two known quantities characterizing the network in this limit,
the NTK and the CK, we have derived several results, from which we
could come up with simple criteria quantifying the desired relationship
between them and the Hamiltonian, criteria which are correlated with
the success of the method, namely that the lower the energy of a Hamiltonian
eigenstate, the larger its kernel expectation should be. This is a
manifestation of a general feature of neural network, which are known
to have implicit biases \citep{https://doi.org/10.48550/arxiv.1905.10264,https://doi.org/10.48550/arxiv.1907.10599}.
We quantified that property using several metrics and showed its relationship
to the network convergence rates and initial state properties, both
theoretically and numerically. 

Using these metrics and theorems, we demonstrated that the overall
performance of the of the method depends both on the basis the Hamiltonian
is presented in and its basis invariant spectrum. Although some models
will be harder than others due to their spectral properties, for example
due to their smaller gap, which is true for almost every numerical
method, the NN method is also very sensitive the chosen basis. Using
the tools developed throughout the paper we were able to characterize
which basis transformation can influence the performance, and demonstrated
how and why common strategies to choose the basis, such as the Marshall-Peierls
sign rule, enhanced performance \citep{Park_2022,Nomura_2021,Bukov_2021}.
We were also able to elucidate the affect of basis dependent properties
of the Hamiltonian, such as its ``stoquasticity'' and the sign structure
of its ground-state. This in turn provided a fundamental explanation
to an empirical observation, that there is a sort of a manifestation
of the sign problem for ANNs, namely, that stoquastic Hamiltonians
are easier than non-stoquastic ones \citep{Bukov_2021,Szabo:2020vk}.
By characterizing and quantifying the basis dependence we also open
a new study direction toward developing methods to find a more suitable
basis or even an optimal one. One may then incorporate those idea
as a sub-routine before optimizing the network. 

Given a Hamiltonian, we also demonstrated how the metrics provided
can predict the overall performance based on the network architecture
and its hyper parameters. Although we limited ourselves to a rather
simple feed-forward network, one can apply the same techniques to
more complicated architectures, recalling the NTK and CK limits apply
to most common ones, such as convolutional neural networks \citep{https://doi.org/10.48550/arxiv.1810.05148,https://doi.org/10.48550/arxiv.1902.04760},
recurrent neural networks \citep{https://doi.org/10.48550/arxiv.2006.10246,https://doi.org/10.48550/arxiv.2104.03093},
transformers \citep{https://doi.org/10.48550/arxiv.2006.10540} and
others \citep{https://doi.org/10.48550/arxiv.2006.14548}. This in
turns might provide a pathway to understanding what are the desirable
characteristics of a neural network architecture when trying to find
the ground state of a given Hamiltonian. 

We also mention that we focused on metrics that provide a more ``macroscopic''
overall description of the relationship between a Hamiltonian and
the kernels. This allowed us to address general questions about general
Hamiltonians. Yet, when one considers a specific model that one wished
to study using a specific ANN, it can be more beneficial to take a
more fine grained approach and define metrics which are relevant to
the task at hand in order to fully describe the relationship between
the specific Hamiltonian and architecture and its influence on the
performance. We leave this as a direction for future work.

While we mainly explored the initialization and dynamics for neural
networks, some of the idea discusses in the paper might be useful
for other ansätze, as we breifly demonstarted in \hyperref[sec:VIII]{VIII}.
As we showed, the sign bias might not manifest itself for some ansätze
such as MPSs and further investigation using similar techniques might
shed a light about other structural differences between different
ansätze. A future study could focus on how the metrics we proposed
might vary between different ansätze and look into their dynamics
in the case of a dynamical NTK. 

To summarize, the techniques developed and used in this paper provide
a way to study both theoretical and practical questions regarding
NQS methods for many-body Hamiltonians. It remains to be seen if and
how similar techniques can be applied to other variational methods,
and how other concepts we disregarded throughout the paper, such as
representability and stochastic optimization strategies, influence
the general conclusions drawn here when a particular setting is considered. 

We would like to thank Y. Bar Sinai, G. Cohen, S. Gazit, and Z. Ringel
for useful discussions. Support by by the Israel Science Foundation
(ISF) and the Directorate for Defense Research and Development (DDR\&D)
Grant No. 3427/21, the ISF grant No. 1113/23, and the US-Israel Binational
Science Foundation (BSF) Grant No. 2020072 is gratefully acknowledged.

\section*{Appendix \protect\label{sec:App}}

\subsection*{App. A Applicability of the NTK limit \protect\label{subsec:App.A}}

It is stated in \citep{10.5555/3454287.3454551} that the NTK limit,
or the lazy training regime as it is often referred to, is guaranteed
to hold whenever 
\[
\frac{\mathcal{L}\left(\theta_{0}\right)}{\Vert\nabla\mathcal{L}\left(\theta_{0}\right)\Vert}\frac{\Vert D^{2}f\left(\theta_{0}\right)\Vert}{\Vert Df\left(\theta_{0}\right)\Vert^{2}}\ll1,
\]
with $\mathcal{L}\left(\theta_{0}\right)$ being the loss function
at initialization, $\nabla$ denoting the natural gradient, i.e, the
gradient with respect to the outputs of the network, $D$ denoting
the differential of the outputs with respect to the parameters, and
$\Vert\cdot\Vert$ denoting the vector or operator norm depending
on the context. It is also assumed without the loss of generality
that the loss is always positive. For the energy loss function the
criterion above amounts to 
\[
\frac{\sqrt{\braket{\psi_{0}\lvert\psi_{0}}}}{\sqrt{\braket{H^{2}}_{0}/\braket{H}_{0}^{2}-1}}\frac{\Vert D^{2}\psi\left(\theta_{0}\right)\Vert}{\Vert D\psi\left(\theta_{0}\right)\Vert^{2}}\ll1,
\]
with $\braket{\cdot}_{0}\equiv\frac{\braket{\psi_{0}\lvert\cdot\lvert\psi_{0}}}{\braket{\psi_{0}\lvert\psi_{0}}}$
being the average over the initial state $\ket{\psi_{0}}$. The term
$\braket{H^{2}}_{0}/\braket{H}_{0}^{2}-1$ approaches some limiting
value in the infinite width limit, which on average depends on the
relationship between the CK and the specific Hamiltonian. This value
is zero if and only if the initial state happens to be an eigenstate
of the Hamiltonian, which is a zero measure event. The term $\braket{\psi_{0}\lvert\psi_{0}}$
approaches the value $\text{Tr}\left(K\right)$ with $K$ being the
CK when the width approaches infinity. As long as one uses the NTK
initialization scheme, as mentioned in \citep{10.5555/3454287.3454551,https://doi.org/10.48550/arxiv.1806.07572},
this term approaches a non zero value. These two terms, hence their
ratio, do not scale with the number of hidden neurons. On the other
hand, the term $\frac{\Vert D^{2}\psi\left(\theta_{0}\right)\Vert}{\Vert D\psi\left(\theta_{0}\right)\Vert^{2}}$
does scales with the number of hidden neurons $N$ as $\mathcal{O}\left(N^{-1/2}\right)$,
as mentioned in \citep{10.5555/3454287.3454551,https://doi.org/10.48550/arxiv.1806.07572}.
Thus, the infinite width limit becomes valid when
\[
N^{-1/2}\ll1.
\]
It is also worth mentioning that contrary to the case of mean-squared
error loss, where one can approach the lazy training regime by rescaling
the output of the neural network by a factor $\alpha\gg1$ \citep{10.5555/3454287.3454551},
the normalization of the energy loss is such that one cannot approach
the lazy training regime by rescaling; indeed, it is easy to see that
the above criterion remains invariant under the transformation $\psi\rightarrow\alpha\psi$.

\subsection*{App. B Proofs for Sec. \hyperref[sec:IV]{IV} \protect\label{subsec:App.B}}
\begin{proof}
Proposition 1: The proof follows immediately from the implicit solution
to Eq. (\ref{eq:8})
\[
\ket{\psi}=\exp\left[\Theta\left(f-Hh\right)\right]\ket{\psi_{0}},
\]
with $f\equiv\int_{0}^{t}\frac{\braket{E}}{\braket{\psi|\psi}}dt'$,
$h\equiv\int_{0}^{t}\frac{1}{\braket{\psi|\psi}}dt'$ and $\ket{\psi_{0}}$
being the initial state. Since $\left[H,\Theta\right]=0$ we have
that 
\[
\frac{\tilde{\psi}_{j}}{\tilde{\psi}_{i}}=\exp\left[\Theta_{j}\left(f-E_{g}h\right)-\Theta_{i}\left(f-E_{i}h\right)\right]\frac{\tilde{\psi}_{j}\left(t=0\right)}{\tilde{\psi}_{i}\left(t=0\right)}.
\]
Taking the square and differentiating with respect to time we get
the result.
\end{proof}
\begin{lem}
Given a $d\times d$ positive definite matrix $A$ with eigenvalues
$a_{i}$, $a_{0}=0$ and constrained Frobenius norm $0<\lvert\lvert A\lvert\lvert_{F}<C$,
the maximal gap $\alpha_{*}=\max_{\alpha_{i}}\min_{i>0}\{\alpha_{i}\}$
occurs when $\alpha_{i}=\frac{C}{\sqrt{d-1}}$. 
\end{lem}
\begin{proof}
Defining the function $\sum e^{-\alpha_{i}t}+\lambda\sum_{i}\alpha_{i}^{2}$,
one can notice that as $t\rightarrow\infty$ the first term is dominated
by $1+e^{-\min_{i>0}\{\alpha_{i}\}t}$, so we wish to minimize this
function. Differentiating with respect to $\alpha_{i}$ for $i>0$
we get $\alpha_{i}=\frac{t}{\lambda}e^{-t\alpha_{i}}$. Using the
constraint we obtain $\lambda^{2}\geq\frac{t^{2}}{C^{2}}\sum_{i>0}e^{-2t\alpha_{i}}$.
Thus, $\alpha_{i}\leq C\frac{e^{-t\alpha_{i}}}{\sqrt{\sum_{i>0}e^{-2t\alpha_{i}}}}$.
Assume that $\alpha_{k}\equiv\min\{\alpha_{i}\}$ satisfies $\alpha_{k}<\alpha_{i}$
for $i\neq k$; then $\alpha_{i}\leq0$ for all $i\neq k$, which
is a contradiction to the assumption that $\alpha_{k}$ is minimal
unless $\alpha_{i}=0\ \forall i$. The only other possibility is that
$\alpha_{i}$ is constant for all $i>0$, which gives $\alpha_{*}\leq\frac{C}{\sqrt{d-1}}$.
Maximizing $\alpha_{*}$ yields the result. 
\end{proof}
\begin{lem}
Given an Hamiltonian $H$, the maximal gap of the matrix $\tilde{H}=\sqrt{\Theta}\left(H-E_{g}\right)\sqrt{\Theta}$
under the constraint that $\lvert\lvert\Theta\lvert\lvert_{F}<C$
for some finite $C$, occurs if and only if $\braket{\Theta}_{\ket{E_{i}}}\propto\frac{1}{E_{i}-E_{g}}$. 
\end{lem}
\begin{proof}
Since the minimal eigenvalue of $\tilde{H}$ is $0$ with eigenvector
$\frac{\sqrt{\Theta^{-1}}\ket{g}}{\bra{g}\Theta^{-1}\ket{g}}$ , using
the lemma above we get that the optimal gap occurs when $\tilde{H}\propto I-\frac{\sqrt{\Theta^{-1}}\ket{g}\bra{g}\sqrt{\Theta^{-1}}}{\bra{g}\Theta^{-1}\ket{g}}$.
Multiplying both sides by $\sqrt{\Theta}$ we obtain $\Theta\left(H-E_{g}\right)\Theta\propto\Theta-\frac{\ket{g}\bra{g}}{\bra{g}\Theta^{-1}\ket{g}}$.
Multiplying then by $H-E_{g}$ we get $\left(H-E_{g}\right)\Theta\left(H-E_{g}\right)\Theta\propto\left(H-E_{g}\right)\Theta$,
which is equivalent to $E\Theta E\propto E$, thus $\left(E_{i}-E_{g}\right)^{2}\bra{E_{i}}\Theta\ket{E_{i}}\propto E_{i}-E_{g}$.
This concludes the proof.
\end{proof}

\subsubsection*{CK Averages Approximations }
\begin{lem}
Consider $\braket{\sigma\lvert\psi}\sim\mathcal{N}\left(0,K\right)$
and $K$ being the CK, with $\braket{\sigma\vert\psi}\in\mathbb{R}^{2^{N}}$;
then
\[
\mathbb{E}\left(\frac{\braket{g\lvert\psi}^{2}}{\braket{\psi\lvert\psi}}\right)=\frac{\mathbb{E}\left(\braket{g\lvert\psi}^{2}\right)}{\mathbb{E}\left(\braket{\psi\lvert\psi}\right)}+\mathcal{O}\left(N\cdot2^{-N/2}\right).
\]
\end{lem}
\begin{proof}
Working with the basis $\ket{s}$, the basis of eigenstates of $K$
\[
\mathbb{E}\left(\frac{\braket{g\lvert\psi}^{2}}{\braket{\psi\lvert\psi}}\right)=\sum_{s}\mathbb{E}\left(\frac{\psi_{s}^{2}}{\braket{\psi\lvert\psi}}\right)\braket{g\lvert s}^{2},
\]
with $\psi_{s}^{2}\equiv\braket{s\lvert\psi}^{2}$. We assume without
loss of generality that the eigenstate with maximal eigenvalue is
the state $\ket{+}$, and that the eigenvalues are normalized such
that $\braket{+\lvert K\lvert+}=1$. Let us denote $\psi_{+}^{2}=\braket{+\vert\psi}^{2}$.
Then, for any $s\neq+$ we have $\text{Cov}\left(\psi_{s}^{2},\braket{\psi\vert\psi}\right)\propto\braket{s\vert K\lvert s}^{4}$,
which, according to the discussion in Sec. \hyperref[sec:III]{III},
decays at least as fast as $\sim\frac{1}{2^{4N}}$. For the case that
$s=+$ we have $\mathbb{E}\left(\frac{\psi_{+}^{2}}{\braket{\psi\lvert\psi}}\right)=\mathbb{E}\left(\frac{\psi_{+}^{2}}{\psi_{+}^{2}+\xi^{2}}\right)$
with $\xi^{2}=\sum_{s\neq+}\psi_{s}^{2}$. Evaluating the Gaussian
integral in the complex plane we get $\mathbb{E}\left(\frac{\psi_{+}^{2}}{\psi_{+}^{2}+\xi^{2}}\right)=1-A\mathbb{E}\left(\xi e^{-\frac{\xi^{2}}{2}}\text{erf}\left(\frac{\xi}{\sqrt{2}}\right)\right)$,
with $A$ being a positive constant. Using Cauchy--Schwarz we get
$\left|\mathbb{E}\left(\frac{\psi_{+}^{2}}{\psi_{+}^{2}+\xi^{2}}\right)-1\right|^{2}\leq B\mathbb{E}\left(\xi^{2}\right)=B\sum_{s\neq+}\braket{s\lvert K\lvert s}$,
with $B$ being a positive constant. Since the $\kappa$th unique
eigenvalue of $K$ decays as $\sim2^{-N\kappa}$, and since the degeneracy
of each subspace is $\binom{N}{\lvert s\lvert}$, using the weak bias
property Eq. (\ref{eq:14}) we have $\sum_{s\neq+}\braket{s\vert K\vert s}\leq2\sum_{k=1}^{N/2}\binom{N}{2k}2^{-Nk}=\left(1+2^{-N/2}\right)^{N}-1$,
which for large $N$ approaches $N\cdot2^{-N/2}$. 
\end{proof}

\subsubsection*{Convergence}
\begin{proof}
Theorem 1: Looking at the energy time derivative we get 
\[
\frac{d\braket{E}}{dt}=-\frac{\bra{\psi}\left(H-\braket{E}\right)\Theta\left(H-\braket{E}\right)\ket{\psi}}{\braket{\psi|\psi}^{2}}.
\]
Following \citep{https://doi.org/10.48550/arxiv.2104.03093} we have
$\Theta>0$ and so $\frac{d\braket{E}}{dt}\leq0$, with equality occurring
only for $\ket{\psi}\propto\ket{E_{i}}$. Since $\braket{E}$ is bounded
(due to $H$ being a bounded operator), convergence is guaranteed.
To see the stability of the fixed point we consider the solution of
the equation near an eigenstate of $H$, $\ket{\psi}=\ket{E_{i}}+\epsilon\ket{\phi}$
with $\epsilon\ll1$, $\braket{g|\phi_{0}}=0$, and $\braket{\phi_{0}|\phi_{0}}=1$,
where $\ket{\phi_{0}}$ is the initial perturbation. Substituting
into Eq. (\ref{eq:8}) we get $\frac{d\ket{\phi}}{dt}=-\Theta\left(H-E_{i}\right)\ket{\phi}+\mathcal{O}\left(\epsilon\right)$.
Solving the equation to first order in $\epsilon$ we find that
\[
\ket{\psi}=\ket{E_{i}}+\epsilon\sqrt{\Theta}e^{-\sqrt{\Theta}\left[H-E_{i}\right]\sqrt{\Theta}t}\sqrt{\Theta^{-1}}\ket{\phi_{0}}+\mathcal{O}\left(\epsilon^{2}\right),
\]
where we assumed that the norm of the initial state near the fixed
point is $1$ to first order in $\epsilon$ without loss of generality.
From the form of the solution the condition for the stability of the
fixed point is that the operator $\sqrt{\Theta}\left(H-E_{i}\right)\sqrt{\Theta}$
will be positive semidefinite. One can see that the expectation of
the operator over the state $\sqrt{\Theta^{-1}}\ket{g}$ is $\frac{E_{g}-E_{i}}{\bra{g}\Theta^{-1}\ket{g}}$,
which is negative except for the case that the fixed point is proportional
to the ground state. 

Theorem 2: To simplify the algebra, we assume here, without loss of
generality, that $E_{g}=0$. \\
Case (1): For $\braket{E}\geq\frac{\Delta}{C}$ we have 
\[
\frac{d\braket{E}}{dt}=-\frac{\braket{\psi|\left(\braket{E}-H\right)\Theta\left(\braket{E}-H\right)|\psi}}{\braket{\psi|\psi}^{2}},
\]
multiplying and dividing by $\braket{g|\Theta^{-1}|g}$ and using
the Cauchy-Schwarz inequality we get 
\[
\frac{d\braket{E}}{dt}\leq-\frac{\braket{\psi|g}^{2}\braket{E}^{2}}{\braket{\psi|\psi}^{2}\braket{g|\Theta^{-1}|g}},
\]
defining $\ket{\tilde{\psi}}$ as the normalized state, $\frac{d\braket{E}}{dt}\leq-\frac{\braket{\tilde{\psi}|g}^{2}\braket{E}\Delta}{\braket{\psi|\psi}\braket{g|\Theta^{-1}|g}C}$.
\\
Using Eq. (\ref{eq:8}) we have $\frac{d\braket{\psi|\Theta^{-1}|\psi}}{dt}=0$,
hence $\braket{\psi|\psi}=\frac{\braket{\psi_{0}|\Theta^{-1}|\psi_{0}}}{\braket{\tilde{\psi}|\Theta^{-1}|\tilde{\psi}}}\leq\Vert\Theta\Vert\braket{\psi_{0}|\Theta^{-1}|\psi_{0}}$
, with $\ket{\psi_{0}}$ the initial state. Assuming without loss
of generality that $\braket{\psi_{0}|\Theta^{-1}|\psi_{0}}=1$ we
obtain 
\[
\frac{d\braket{E}}{dt}\leq-\frac{\min\left(\braket{\tilde{\psi}|g}^{2}\right)\Delta\braket{E}}{C\Vert\Theta\Vert\braket{g|\Theta^{-1}|g}},
\]
which concludes the proof of the first bound. \\
For the case $\braket{E}\leq\frac{\Delta}{C}$ we have 
\[
\frac{d\braket{E}}{dt}\leq-\frac{\braket{\psi|g}^{2}\braket{E}^{2}}{2\braket{\psi|\psi}^{2}\braket{g|\Theta^{-1}|g}}-\frac{1}{2d-2}\sum_{i>0}^{d}\frac{\braket{\psi|\left(\braket{E}-H\right)\Theta\left(\braket{E}-H\right)|\psi}}{\braket{\psi|\psi}^{2}}\frac{\braket{E_{i}|\Theta^{-1}|E_{i}}}{\braket{E_{i}|\Theta^{-1}|E_{i}}},
\]
where $d\equiv2^{N}$ and $\ket{E_{i}}$ are the eigenstate of $H$.
Using Cauchy-Schwarz we have 
\[
\frac{d\braket{E}}{dt}\leq-\frac{\braket{\psi|g}^{2}\braket{E}^{2}}{2\braket{\psi|\psi}^{2}\braket{g|\Theta^{-1}|g}}-\frac{1}{2d-2}\sum_{i>0}\frac{\braket{\psi|E_{i}}^{2}\left(\braket{E}-E_{i}\right)^{2}}{\braket{\psi|\psi}^{2}\braket{E_{i}|\Theta^{-1}|E_{i}}}.
\]
Focusing on the sum on the RHS, we have $\sum_{i>0}\frac{\braket{\psi|E_{i}}^{2}\left(\braket{E}-E_{i}\right)^{2}}{\braket{\psi|\psi}^{2}\braket{E_{i}|\Theta^{-1}|E_{i}}}\geq\frac{1}{\Vert\Theta\Vert}\min_{i\geq1}\left(\frac{\lvert E_{i}-\Delta/C\lvert}{\braket{i|\Theta^{-1}|i}}\right)\sum_{i>0}\tilde{\psi}_{i}\left(E_{i}-\braket{E}\right)=\frac{1}{\Vert\Theta\Vert}\min_{i\geq1}\left(\frac{\lvert E_{i}-\Delta/C\lvert}{\braket{i|\Theta^{-1}|i}}\right)\tilde{\psi}_{g}^{2}\braket{E}$.
Thus, overall we get $\frac{d\braket{E}}{dt}\leq-\frac{a}{2}\braket{E}^{2}-b\braket{E}$
with $a,b$ defined in Theorem 2. Solving the inequality we arrive
at the final result. \\
Case (2): Using Eq. (\ref{eq:8}), 
\[
\frac{d\psi_{g}}{dt}=\frac{\braket{g|\Theta|g}}{\braket{\psi|\Theta|\psi}}\braket{E}\psi_{g}\geq\frac{\braket{g|\Theta|g}}{\Vert\Theta\Vert}\frac{\Delta}{C}\psi_{g},
\]
 due to $\ket{g}$ being an eigenstate of $\Theta$. We thus find
\[
\psi_{g}\geq\psi_{g}\left(t=0\right)\exp\left[\frac{\braket{g|\Theta|g}}{\Vert\Theta\Vert}\frac{\Delta}{C}t\right],
\]
which implies that 
\[
\tilde{\psi}_{g}^{2}\geq\frac{\Theta_{\min}}{\Vert\Theta\Vert}\tilde{\psi}_{g}\left(t=0\right)\exp\left[\frac{\braket{g|\Theta|g}}{\Vert\Theta\Vert}\frac{\Delta}{C}t\right],
\]
since $\min\left(\braket{\psi|\psi}\right)\geq\Theta_{\min}$. One
can then see that $\braket{E}=\sum_{i}\tilde{\psi}_{i}^{2}E_{i}\leq\sum_{i>0}\psi_{i}^{2}\Vert H\Vert=\left(1-\tilde{\psi}_{g}^{2}\right)\Vert H\Vert$
and employ the fact that $1-x<\frac{1}{x}$ for every $x>0$ to get
the final bound.
\[
\]
\\
Case (3): Let us start from $\frac{d\braket{E}}{dt}=-\frac{\braket{\psi|\left(\braket{E}-H\right)\Theta\left(\braket{E}-H\right)|\psi}}{\braket{\psi|\psi}^{2}}$.
Since $\left[\Theta,H\right]=0$ we have $\frac{d\braket{E}}{dt}=-\frac{1}{\braket{\psi|\psi}}\sum_{i}\psi_{i}^{2}\left(\braket{E}-E_{i}\right)^{2}\Theta_{i}\leq-\frac{\braket{\psi|g}^{2}\braket{E}^{2}\Theta_{g}}{2\braket{\psi|\psi}^{2}}-\sum_{i>0}\frac{\braket{\psi|E_{i}}^{2}\left(\braket{E}-E_{i}\right)^{2}\Theta_{i}}{\braket{\psi|\psi}^{2}}$
. Continuing the argument as in Case (1) we obtain the result. 

Proposition 2: Let us examine $\frac{d\psi_{g}}{dt}=\frac{1}{\braket{\psi|\psi}}\bra{g}\Theta\left(\braket{E}-H\right)\ket{\psi}=\frac{1}{\braket{\psi|\psi}}\bra{g}\Theta\left(P_{g}+P_{g}^{\perp}\right)\left(\braket{E}-H\right)\ket{\psi}$.
If $\frac{1}{2}\bra{g}\Theta P_{g}\left(\braket{E}-H\right)\ket{\psi}\geq\bra{g}\Theta P_{g}^{\perp}\left(\braket{E}-H\right)\ket{\psi}$
then $\frac{d\psi_{g}}{dt}\geq\frac{1}{2\braket{\psi|\psi}}\braket{g|\Theta|g}\braket{E}\psi_{g}$
and Case (2) argument follows. In order to satisfy the inequality
it is enough to satisfy $\frac{1}{4}\braket{g|\Theta|g}^{2}\tilde{\psi}_{g}^{2}\braket{E}^{2}\geq\braket{g|\Theta P_{g}^{\perp}\Theta|g}\left(\braket{E^{2}}-\braket{E}^{2}\right)$,
which is the same as $\frac{1}{4}\frac{\braket{g|\Theta|g}^{2}}{\braket{g|\Theta^{2}|g}-\braket{g|\Theta|g}^{2}}\geq\frac{1}{\tilde{\psi}_{g}^{2}}\left(\frac{\braket{E^{2}}}{\braket{E}^{2}}-1\right)$.
We can now use $\frac{1}{\tilde{\psi}_{g}^{2}}\left(\frac{\braket{E^{2}}}{\braket{E}^{2}}-1\right)\leq\frac{1}{\tilde{\psi}_{g}^{2}}\frac{C\Vert H\Vert}{\Delta}$.
Moreover, if $\frac{d\psi_{g}^{2}}{dt}\geq0$ then $\tilde{\psi_{g}^{2}}\geq\frac{\psi_{g}^{2}}{\Vert\Theta\Vert}$.
It is therefore easy to see that satisfying the inequality for $t=0$
is then enough for it to hold for all $t\geq0$. \\

Theorem 3: In order to converge, the state must be, at some point
in time, arbitrarily close to the ground state. We use the approximation
$\ket{\psi}=N\left(\ket{g}+\epsilon\ket{\phi}\right)$, where $\ket{\phi}$
is orthogonal to the ground state and its norm is less than unity
and $\epsilon\ll1$. $N$ is the norm, $N=\frac{\braket{\psi_{0}|\Theta^{-1}|\psi_{0}}}{\braket{\psi|\Theta^{-1}|\psi}}$.
Writing down Eq. (\ref{eq:8}) in the discrete time form, under those
conditions we find to first order in $\epsilon$ that 
\[
\ket{\psi^{t+1}}=\left[\mathbb{I}-\Delta t\frac{\braket{g|\Theta^{-1}|g}}{\braket{\psi_{0}|\Theta^{-1}|\psi_{0}}}\Theta\left[H-E_{g}\mathbb{I}\right]\right]\ket{\psi^{t}}+\mathcal{O}\left(\epsilon\right).
\]
After $n$ steps from some time $t$ we have 
\[
\ket{\psi^{t+n}}=\left[\mathbb{I}-\Delta t\frac{\braket{g|\Theta^{-1}|g}}{\braket{\psi_{0}|\Theta^{-1}|\psi_{0}}}\Theta\left[H-E_{g}\mathbb{I}\right]\right]^{n}\ket{\psi^{t}}+\mathcal{O}\left(\epsilon\right),
\]
and so one has that all the eigenvalues of the operator $\mathbb{I}-\Delta t\frac{\braket{g|\Theta^{-1}|g}}{\braket{\psi_{0}|\Theta^{-1}|\psi_{0}}}\Theta\left[H-E_{g}\mathbb{I}\right]$
must be smaller than unity in absolute value. A necessary condition
for that to hold for all the eigenvalues can be seen to be $\Delta t\leq\frac{\braket{\psi_{0}|\Theta^{-1}|\psi_{0}}}{\braket{g|\Theta^{-1}|g}}\frac{2}{\left\Vert \sqrt{\Theta}\left(H-E_{g}\right)\sqrt{\Theta}\right\Vert }.$
\end{proof}

\subsection*{App. C Proofs for Sec. \hyperref[sec:V]{V} \protect\label{subsec:App.C}}
\begin{proof}
Theorem 4: Assume that $\ket{\psi}$ is a solution of Eq. (\ref{eq:8})
with a Hamiltonian $H$ and initial condition $\ket{\psi_{0}}$, so
$\frac{d\ket{\psi}}{dt}=\frac{1}{\braket{\psi|\psi}}\Theta\left[\braket{E}-H\right]\ket{\psi}$.
Multiplying both sides by $U$ and defining $\ket{\phi}\equiv U\ket{\psi}$,
and $\tilde{H}\equiv UHU^{\dagger}$, we have $\frac{d\ket{\phi}}{dt}=\frac{1}{\braket{\phi|\phi}}U\Theta U^{\dagger}U\left[\braket{\tilde{E}}-H\right]U^{\dagger}U\ket{\psi}$,
where we also used $U^{\dagger}U=I$, so $\braket{\phi|\phi}=\braket{\psi|\psi}$
and $\braket{E}=\frac{\braket{\psi|H|\psi}}{\braket{\psi|\psi}}=\frac{\braket{\phi|H|\phi}}{\braket{\phi|\phi}}$.
Since $\left[\Theta,U\right]=0$ we get $\frac{d\ket{\phi}}{dt}=\frac{1}{\braket{\phi|\phi}}\Theta\left[\braket{\tilde{E}}-\tilde{H}\right]\ket{\phi}$.
\\
The average initial energy is $\mathbb{E}\left(\braket{E_{0}}\right)=\int\frac{\braket{\psi_{0}|H|\psi_{0}}}{\braket{\psi_{0}|\psi_{0}}}\exp\left[-\frac{1}{2}\braket{\psi_{0}|K^{-1}|\psi_{0}}\right]d\ket{\psi_{0}}$
and the initial average overlap is $\mathbb{E}\left(\tilde{\psi}_{g}^{2}(t=0)\right)=\int\frac{\braket{g|\psi_{0}}^{2}}{\braket{\psi_{0}|\psi_{0}}}\exp\left[-\frac{1}{2}\braket{\psi_{0}|K^{-1}|\psi_{0}}\right]d\ket{\psi_{0}}$.
Changing the integration variable from $\ket{\psi}$ to $\ket{\phi}$
and using the fact that the Jacobian of the transformation is 1, we
obtain the result. 

Corollary 1: Since the general form of the NTK and CK is $\Theta=\sum_{S}\alpha_{S}X_{S}$
with $S\subset\left\{ 1,...,N\right\} $ and $\alpha_{S}$ being real,
then $\left[\Theta,e^{i\chi}\right]=0$ for every $\chi=\sum_{S}\beta_{S}X_{S}$.

Corollary 2: As explained in the main text, altering the signs of
$\beta_{S}$ and $\gamma_{S}$ without altering $\alpha_{S}$ and
the magnitudes of $\beta_{S},\gamma_{S}$ can only be done using a
unitary of the form $X_{S}$, which will always commutes with the
NTK and CK.

Proposition 3: \\
(1) Due to the general form of $\Theta$ and $K$, $\text{Tr}\left(\Theta\mathcal{H}\right)=\sum\text{Tr}\left(\Theta X_{S}\right)\alpha_{S}$,
where we used the orthogonality properties of the Pauli matrices.
Since $\Theta$ and $K$ are assumed to be entry-wise positive, $\text{Tr}\left(\Theta X_{S}\right)\geq0$,
and so if one can only change the signs of $\alpha_{s}$, the minimum
will be obtained whenever $\text{sgn}\left(\alpha_{s}\right)=-1$,
provided that it is attainable.\\
(2) More formally, the statement can be formulated as: If there is
a unitary transformation $V\in\left\{ Z_{S}\right\} $ that makes
$\alpha_{s}\leq0$, then $V=\text{argmax}_{U\in\left\{ Z_{S}\right\} }\left\{ \braket{g|U^{\dagger}\Theta U|g}\right\} $.
Writing the ground state of the Hamiltonian in the computational basis
$\ket{g}=\sum_{\sigma}g_{\sigma}\ket{\sigma}$, sign transformations
of the form $Z_{S}$ will only change the signs of $g_{\sigma}$ but
not the absolute values $\lvert g_{\sigma}\lvert$. $\braket{g|\Theta|g}=\sum_{\sigma\sigma'}\Theta_{\sigma,\sigma'}g_{\sigma}g_{\sigma'}\leq\sum_{\sigma\sigma'}\Theta_{\sigma,\sigma'}\lvert g_{\sigma}\lvert\lvert g_{\sigma'}\lvert$,
due to $\Theta_{\sigma,\sigma'}\geq0$, which implies that the upper
bound is obtained if and only if $g_{\sigma}\geq0$. This situations
occurs if and only if the Hamiltonian is stoquastic, which, under
the assumptions, occurs only if $\alpha_{S}\leq0$.\\
(3) In the extremely biased case, since $\Theta_{\sigma,\sigma'}\geq0$,
we have $\Theta\approx\ket{+}\bra{+}+\mathcal{O}\left(\epsilon\right)$
with $\epsilon\ll1$. Then the learning rate bound can be approximately
written as $\frac{\epsilon}{1-\left(1-\epsilon\right)\braket{g|+}^{2}}\frac{1}{\braket{+|H|+}-E_{g}}$.
The first term is a function of $\sum_{\sigma}g_{\sigma}$, hence
it is maximized by the case $g_{s}\geq0\ \forall s$ (or $g_{s}\leq0\ \forall s$),
which only occurs if $\alpha_{s}\leq0$. As for the second term, one
can see that $\braket{+|H|+}=\text{Tr\ensuremath{\left(H\ket{+}\bra{+}\right)}=\ensuremath{\sum_{s}\text{Tr}\left(HX_{s}\right)=\sum}\ensuremath{\alpha_{s}}}$,
which is minimized by $\alpha_{s}\leq0$. Thus the whole expression
is maximized by having $\alpha_{s}\leq0$. 

Proposition (4): This statement follows from $\braket{g|\Theta|g}=\sum_{\sigma,\sigma'}g_{\sigma}g_{\sigma'}\Theta_{\sigma,\sigma'}\leq\sum_{\sigma,\sigma'}\lvert g_{\sigma}\lvert\lvert g_{\sigma'}\lvert\Theta_{\sigma,\sigma'}$.
Without changing the norms the bound is obtained if and only if $g_{\sigma}$
have the same signs for all $\sigma$.

Proposition (5): As written in the main text, $\text{Tr}\left(H_{2}\Theta\right)=\text{\ensuremath{\sum_{i,j}J_{i,j}^{(X,X)}\alpha_{2}}}+h_{i}^{(X)}\alpha_{1}$,
with $\alpha_{1},\alpha_{2}\geq0$, since $\Theta$ is entry-wise
positive. Hence, $\text{Tr}\left(H_{2}\Theta\right)\geq-\text{\ensuremath{\sum_{i,j}\lvert J_{i,j}^{(X,X)}\lvert\alpha_{2}}}+\lvert h_{i}^{(X)}\lvert\alpha_{1}$.
\end{proof}

\subsection*{App. D Details of the numerical calculations \protect\label{subsec:App.D}}

Throughout the paper we used the common NTK initialization strategy
as defined in \citep{https://doi.org/10.48550/arxiv.1907.10599} with
a zero bias initialization. 

For each plot with one layer we used $5\cdot10^{4}$ hidden units,
and for ones with more or varying number of layers we used $2\cdot10^{3}$,
which is still much larger than the number of degrees of freedom.
For benchmarks of the convergence times we used systems with $N=5$,
which although small, already allows us to see the predicted trends.
For every figure addressing the average initialization we used $N=9$,
except for the case where we vary the number of qubits, where we went
to $N=11$. 

The convergence times were measured by the time it takes for the energy
or ground state overlap to reach a certain threshold $\epsilon$.
For the ground state this time is defined as the smallest time for
which $\tilde{\psi}_{g}\geq0.95$, while for the energy it is the
smallest time for which $\braket{E}-E_{g}\leq\frac{\Delta}{C}$, with
$\Delta$ being the energy gap and $C=4$.

Both the NTK and CK were calculated by averaging over initializations.
In order to determine the number of samples and hidden units we calculated
the relative error of the eigenvalues of each kernel, which is the
standard deviation over the mean. We chose a number of samples and
hidden units such that the error is less than $1\%$ for all of the
eigenvalues.

\bibliographystyle{apsrev4-1_c}
\bibliography{ref}

\begin{thebibliography}{42}%
\makeatletter
\providecommand \@ifxundefined [1]{%
 \@ifx{#1\undefined}
}%
\providecommand \@ifnum [1]{%
 \ifnum #1\expandafter \@firstoftwo
 \else \expandafter \@secondoftwo
 \fi
}%
\providecommand \@ifx [1]{%
 \ifx #1\expandafter \@firstoftwo
 \else \expandafter \@secondoftwo
 \fi
}%
\providecommand \natexlab [1]{#1}%
\providecommand \enquote  [1]{``#1''}%
\providecommand \bibnamefont  [1]{#1}%
\providecommand \bibfnamefont [1]{#1}%
\providecommand \citenamefont [1]{#1}%
\providecommand \href@noop [0]{\@secondoftwo}%
\providecommand \href [0]{\begingroup \@sanitize@url \@href}%
\providecommand \@href[1]{\@@startlink{#1}\@@href}%
\providecommand \@@href[1]{\endgroup#1\@@endlink}%
\providecommand \@sanitize@url [0]{\catcode `\\12\catcode `\$12\catcode
  `\&12\catcode `\#12\catcode `\^12\catcode `\_12\catcode `\%12\relax}%
\providecommand \@@startlink[1]{}%
\providecommand \@@endlink[0]{}%
\providecommand \url  [0]{\begingroup\@sanitize@url \@url }%
\providecommand \@url [1]{\endgroup\@href {#1}{\urlprefix }}%
\providecommand \urlprefix  [0]{URL }%
\providecommand \Eprint [0]{\href }%
\providecommand \doibase [0]{http://dx.doi.org/}%
\providecommand \selectlanguage [0]{\@gobble}%
\providecommand \bibinfo  [0]{\@secondoftwo}%
\providecommand \bibfield  [0]{\@secondoftwo}%
\providecommand \translation [1]{[#1]}%
\providecommand \BibitemOpen [0]{}%
\providecommand \bibitemStop [0]{}%
\providecommand \bibitemNoStop [0]{.\EOS\space}%
\providecommand \EOS [0]{\spacefactor3000\relax}%
\providecommand \BibitemShut  [1]{\csname bibitem#1\endcsname}%
\let\auto@bib@innerbib\@empty
\bibitem [{\citenamefont {Kempe}\ \emph {et~al.}(2006)\citenamefont {Kempe},
  \citenamefont {Kitaev},\ and\ \citenamefont
  {Regev}}]{https://doi.org/10.48550/arxiv.quant-ph/0406180}%
  \BibitemOpen
  \bibfield  {author} {\bibinfo {author} {\bibfnamefont {J.}~\bibnamefont
  {Kempe}}, \bibinfo {author} {\bibfnamefont {A.}~\bibnamefont {Kitaev}}, \
  and\ \bibinfo {author} {\bibfnamefont {O.}~\bibnamefont {Regev}},\ }\bibfield
   {title} {\emph {\enquote {\bibinfo {title} {The complexity of the local
  hamiltonian problem},}\ }}\bibfield  {journal} {\bibinfo  {journal} {SIAM
  Journal on Computing}\ }\textbf {\bibinfo {volume} {35}},\ \bibinfo {pages}
  {1070} (\bibinfo {year} {2006})\BibitemShut {NoStop}%
\bibitem [{\citenamefont {Hastings}(2006)}]{Hastings_2006}%
  \BibitemOpen
  \bibfield  {author} {\bibinfo {author} {\bibfnamefont {M.~B.}\ \bibnamefont
  {Hastings}},\ }\bibfield  {title} {\emph {\enquote {\bibinfo {title} {Solving
  gapped hamiltonians locally},}\ }}\bibfield  {journal} {\bibinfo  {journal}
  {Physical Review B}\ }\textbf {\bibinfo {volume} {73}},\ \bibinfo {pages}
  {085115} (\bibinfo {year} {2006})\BibitemShut {NoStop}%
\bibitem [{\citenamefont {Cirac}\ \emph {et~al.}(2021)\citenamefont {Cirac},
  \citenamefont {P{\'{e}}rez-Garc{\'{\i}}a}, \citenamefont {Schuch},\ and\
  \citenamefont {Verstraete}}]{Cirac_2021}%
  \BibitemOpen
  \bibfield  {author} {\bibinfo {author} {\bibfnamefont {J.~I.}\ \bibnamefont
  {Cirac}}, \bibinfo {author} {\bibfnamefont {D.}~\bibnamefont
  {P{\'{e}}rez-Garc{\'{\i}}a}}, \bibinfo {author} {\bibfnamefont
  {N.}~\bibnamefont {Schuch}}, \ and\ \bibinfo {author} {\bibfnamefont
  {F.}~\bibnamefont {Verstraete}},\ }\bibfield  {title} {\emph {\enquote
  {\bibinfo {title} {Matrix product states and projected entangled pair states:
  Concepts, symmetries, theorems},}\ }}\bibfield  {journal} {\bibinfo
  {journal} {Reviews of Modern Physics}\ }\textbf {\bibinfo {volume} {93}},\
  \bibinfo {pages} {045003} (\bibinfo {year} {2021})\BibitemShut {NoStop}%
\bibitem [{\citenamefont {Carleo}\ and\ \citenamefont
  {Troyer}(2017)}]{Carleo_2017}%
  \BibitemOpen
  \bibfield  {author} {\bibinfo {author} {\bibfnamefont {G.}~\bibnamefont
  {Carleo}}\ and\ \bibinfo {author} {\bibfnamefont {M.}~\bibnamefont
  {Troyer}},\ }\bibfield  {title} {\emph {\enquote {\bibinfo {title} {Solving
  the quantum many-body problem with artificial neural networks},}\ }}\bibfield
   {journal} {\bibinfo  {journal} {Science}\ }\textbf {\bibinfo {volume}
  {355}},\ \bibinfo {pages} {602} (\bibinfo {year} {2017})\BibitemShut
  {NoStop}%
\bibitem [{\citenamefont {Gao}\ and\ \citenamefont {Duan}(2017)}]{Gao_2017}%
  \BibitemOpen
  \bibfield  {author} {\bibinfo {author} {\bibfnamefont {X.}~\bibnamefont
  {Gao}}\ and\ \bibinfo {author} {\bibfnamefont {L.-M.}\ \bibnamefont {Duan}},\
  }\bibfield  {title} {\emph {\enquote {\bibinfo {title} {Efficient
  representation of quantum many-body states with deep neural networks},}\
  }}\bibfield  {journal} {\bibinfo  {journal} {Nature Communications}\ }\textbf
  {\bibinfo {volume} {8}},\ \bibinfo {pages} {662} (\bibinfo {year}
  {2017})\BibitemShut {NoStop}%
\bibitem [{\citenamefont {Torlai}\ \emph {et~al.}(2018)\citenamefont {Torlai},
  \citenamefont {Mazzola}, \citenamefont {Carrasquilla}, \citenamefont
  {Troyer}, \citenamefont {Melko},\ and\ \citenamefont {Carleo}}]{Torlai_2018}%
  \BibitemOpen
  \bibfield  {author} {\bibinfo {author} {\bibfnamefont {G.}~\bibnamefont
  {Torlai}}, \bibinfo {author} {\bibfnamefont {G.}~\bibnamefont {Mazzola}},
  \bibinfo {author} {\bibfnamefont {J.}~\bibnamefont {Carrasquilla}}, \bibinfo
  {author} {\bibfnamefont {M.}~\bibnamefont {Troyer}}, \bibinfo {author}
  {\bibfnamefont {R.}~\bibnamefont {Melko}}, \ and\ \bibinfo {author}
  {\bibfnamefont {G.}~\bibnamefont {Carleo}},\ }\bibfield  {title} {\emph
  {\enquote {\bibinfo {title} {Neural-network quantum state tomography},}\
  }}\bibfield  {journal} {\bibinfo  {journal} {Nature Physics}\ }\textbf
  {\bibinfo {volume} {14}},\ \bibinfo {pages} {447} (\bibinfo {year}
  {2018})\BibitemShut {NoStop}%
\bibitem [{\citenamefont {Cubitt}\ \emph {et~al.}(2018)\citenamefont {Cubitt},
  \citenamefont {Montanaro},\ and\ \citenamefont {Piddock}}]{Cubitt_2018}%
  \BibitemOpen
  \bibfield  {author} {\bibinfo {author} {\bibfnamefont {T.~S.}\ \bibnamefont
  {Cubitt}}, \bibinfo {author} {\bibfnamefont {A.}~\bibnamefont {Montanaro}}, \
  and\ \bibinfo {author} {\bibfnamefont {S.}~\bibnamefont {Piddock}},\
  }\bibfield  {title} {\emph {\enquote {\bibinfo {title} {Universal quantum
  hamiltonians},}\ }}\bibfield  {journal} {\bibinfo  {journal} {Proceedings of
  the National Academy of Sciences}\ }\textbf {\bibinfo {volume} {115}},\
  \bibinfo {pages} {9497} (\bibinfo {year} {2018})\BibitemShut {NoStop}%
\bibitem [{\citenamefont {Luo}\ and\ \citenamefont
  {Halverson}(2023)}]{Luo_2023}%
  \BibitemOpen
  \bibfield  {author} {\bibinfo {author} {\bibfnamefont {D.}~\bibnamefont
  {Luo}}\ and\ \bibinfo {author} {\bibfnamefont {J.}~\bibnamefont
  {Halverson}},\ }\bibfield  {title} {\emph {\enquote {\bibinfo {title}
  {Infinite neural network quantum states: entanglement and training
  dynamics},}\ }}\bibfield  {journal} {\bibinfo  {journal} {Machine Learning:
  Science and Technology}\ }\textbf {\bibinfo {volume} {4}},\ \bibinfo {pages}
  {025038} (\bibinfo {year} {2023})\BibitemShut {NoStop}%
\bibitem [{\citenamefont {Liang}\ \emph {et~al.}(2018)\citenamefont {Liang},
  \citenamefont {Liu}, \citenamefont {Lin}, \citenamefont {Guo}, \citenamefont
  {Zhang},\ and\ \citenamefont {He}}]{Liang_2018}%
  \BibitemOpen
  \bibfield  {author} {\bibinfo {author} {\bibfnamefont {X.}~\bibnamefont
  {Liang}}, \bibinfo {author} {\bibfnamefont {W.-Y.}\ \bibnamefont {Liu}},
  \bibinfo {author} {\bibfnamefont {P.-Z.}\ \bibnamefont {Lin}}, \bibinfo
  {author} {\bibfnamefont {G.-C.}\ \bibnamefont {Guo}}, \bibinfo {author}
  {\bibfnamefont {Y.-S.}\ \bibnamefont {Zhang}}, \ and\ \bibinfo {author}
  {\bibfnamefont {L.}~\bibnamefont {He}},\ }\bibfield  {title} {\emph {\enquote
  {\bibinfo {title} {Solving frustrated quantum many-particle models with
  convolutional neural networks},}\ }}\bibfield  {journal} {\bibinfo  {journal}
  {Physical Review B}\ }\textbf {\bibinfo {volume} {98}},\ \bibinfo {pages}
  {104426} (\bibinfo {year} {2018})\BibitemShut {NoStop}%
\bibitem [{\citenamefont {Hibat-Allah}\ \emph {et~al.}(2020)\citenamefont
  {Hibat-Allah}, \citenamefont {Ganahl}, \citenamefont {Hayward}, \citenamefont
  {Melko},\ and\ \citenamefont {Carrasquilla}}]{Hibat_Allah_2020}%
  \BibitemOpen
  \bibfield  {author} {\bibinfo {author} {\bibfnamefont {M.}~\bibnamefont
  {Hibat-Allah}}, \bibinfo {author} {\bibfnamefont {M.}~\bibnamefont {Ganahl}},
  \bibinfo {author} {\bibfnamefont {L.~E.}\ \bibnamefont {Hayward}}, \bibinfo
  {author} {\bibfnamefont {R.~G.}\ \bibnamefont {Melko}}, \ and\ \bibinfo
  {author} {\bibfnamefont {J.}~\bibnamefont {Carrasquilla}},\ }\bibfield
  {title} {\emph {\enquote {\bibinfo {title} {Recurrent neural network wave
  functions},}\ }}\bibfield  {journal} {\bibinfo  {journal} {Physical Review
  Research}\ }\textbf {\bibinfo {volume} {2}},\ \bibinfo {pages} {023358}
  (\bibinfo {year} {2020})\BibitemShut {NoStop}%
\bibitem [{\citenamefont {Sharir}\ \emph {et~al.}(2020)\citenamefont {Sharir},
  \citenamefont {Levine}, \citenamefont {Wies}, \citenamefont {Carleo},\ and\
  \citenamefont {Shashua}}]{Sharir_2020}%
  \BibitemOpen
  \bibfield  {author} {\bibinfo {author} {\bibfnamefont {O.}~\bibnamefont
  {Sharir}}, \bibinfo {author} {\bibfnamefont {Y.}~\bibnamefont {Levine}},
  \bibinfo {author} {\bibfnamefont {N.}~\bibnamefont {Wies}}, \bibinfo {author}
  {\bibfnamefont {G.}~\bibnamefont {Carleo}}, \ and\ \bibinfo {author}
  {\bibfnamefont {A.}~\bibnamefont {Shashua}},\ }\bibfield  {title} {\emph
  {\enquote {\bibinfo {title} {Deep autoregressive models for the efficient
  variational simulation of many-body quantum systems},}\ }}\bibfield
  {journal} {\bibinfo  {journal} {Physical Review Letters}\ }\textbf {\bibinfo
  {volume} {124}},\ \bibinfo {pages} {020503} (\bibinfo {year}
  {2020})\BibitemShut {NoStop}%
\bibitem [{\citenamefont {Westerhout}\ \emph {et~al.}(2020)\citenamefont
  {Westerhout}, \citenamefont {Astrakhantsev}, \citenamefont {Tikhonov},
  \citenamefont {Katsnelson},\ and\ \citenamefont {Bagrov}}]{Westerhout_2020}%
  \BibitemOpen
  \bibfield  {author} {\bibinfo {author} {\bibfnamefont {T.}~\bibnamefont
  {Westerhout}}, \bibinfo {author} {\bibfnamefont {N.}~\bibnamefont
  {Astrakhantsev}}, \bibinfo {author} {\bibfnamefont {K.~S.}\ \bibnamefont
  {Tikhonov}}, \bibinfo {author} {\bibfnamefont {M.~I.}\ \bibnamefont
  {Katsnelson}}, \ and\ \bibinfo {author} {\bibfnamefont {A.~A.}\ \bibnamefont
  {Bagrov}},\ }\bibfield  {title} {\emph {\enquote {\bibinfo {title}
  {Generalization properties of neural network approximations to frustrated
  magnet ground states},}\ }}\bibfield  {journal} {\bibinfo  {journal} {Nature
  Communications}\ }\textbf {\bibinfo {volume} {11}},\ \bibinfo {pages} {1593}
  (\bibinfo {year} {2020})\BibitemShut {NoStop}%
\bibitem [{\citenamefont {Szab{\'o}}(2020)}]{Szabo:2020vk}%
  \BibitemOpen
  \bibfield  {author} {\bibinfo {author} {\bibfnamefont {A.}~\bibnamefont
  {Szab{\'o}}},\ }\bibfield  {title} {\emph {\enquote {\bibinfo {title} {Neural
  network wave functions and the sign problem},}\ }}\bibfield  {journal}
  {\bibinfo  {journal} {Physical Review Research}\ }\textbf {\bibinfo {volume}
  {2}},\ \bibinfo {pages} {033075} (\bibinfo {year} {2020})\BibitemShut
  {NoStop}%
\bibitem [{\citenamefont {Park}\ and\ \citenamefont
  {Kastoryano}(2022)}]{Park_2022}%
  \BibitemOpen
  \bibfield  {author} {\bibinfo {author} {\bibfnamefont {C.-Y.}\ \bibnamefont
  {Park}}\ and\ \bibinfo {author} {\bibfnamefont {M.~J.}\ \bibnamefont
  {Kastoryano}},\ }\bibfield  {title} {\emph {\enquote {\bibinfo {title}
  {Expressive power of complex-valued restricted boltzmann machines for solving
  nonstoquastic hamiltonians},}\ }}\bibfield  {journal} {\bibinfo  {journal}
  {Physical Review B}\ }\textbf {\bibinfo {volume} {106}},\ \bibinfo {pages}
  {134437} (\bibinfo {year} {2022})\BibitemShut {NoStop}%
\bibitem [{\citenamefont {Bukov}\ \emph {et~al.}(2021)\citenamefont {Bukov},
  \citenamefont {Schmitt},\ and\ \citenamefont {Dupont}}]{Bukov_2021}%
  \BibitemOpen
  \bibfield  {author} {\bibinfo {author} {\bibfnamefont {M.}~\bibnamefont
  {Bukov}}, \bibinfo {author} {\bibfnamefont {M.}~\bibnamefont {Schmitt}}, \
  and\ \bibinfo {author} {\bibfnamefont {M.}~\bibnamefont {Dupont}},\
  }\bibfield  {title} {\emph {\enquote {\bibinfo {title} {Learning the ground
  state of a non-stoquastic quantum hamiltonian in a rugged neural network
  landscape},}\ }}\bibfield  {journal} {\bibinfo  {journal} {{SciPost}
  Physics}\ }\textbf {\bibinfo {volume} {10}},\ \bibinfo {pages} {147}
  (\bibinfo {year} {2021})\BibitemShut {NoStop}%
\bibitem [{\citenamefont {Nomura}(2021)}]{Nomura_2021}%
  \BibitemOpen
  \bibfield  {author} {\bibinfo {author} {\bibfnamefont {Y.}~\bibnamefont
  {Nomura}},\ }\bibfield  {title} {\emph {\enquote {\bibinfo {title} {Helping
  restricted boltzmann machines with quantum-state representation by restoring
  symmetry},}\ }}\bibfield  {journal} {\bibinfo  {journal} {Journal of Physics:
  Condensed Matter}\ }\textbf {\bibinfo {volume} {33}},\ \bibinfo {pages}
  {174003} (\bibinfo {year} {2021})\BibitemShut {NoStop}%
\bibitem [{\citenamefont {Jacot}\ \emph {et~al.}(2018)\citenamefont {Jacot},
  \citenamefont {Gabriel},\ and\ \citenamefont
  {Hongler}}]{https://doi.org/10.48550/arxiv.1806.07572}%
  \BibitemOpen
  \bibfield  {author} {\bibinfo {author} {\bibfnamefont {A.}~\bibnamefont
  {Jacot}}, \bibinfo {author} {\bibfnamefont {F.}~\bibnamefont {Gabriel}}, \
  and\ \bibinfo {author} {\bibfnamefont {C.}~\bibnamefont {Hongler}},\ }in\
  \href {\doibase 10.48550/ARXIV.1806.07572} {\emph {\bibinfo {booktitle}
  {Proceedings of the 32nd International Conference on Neural Information
  Processing Systems}}},\ \bibinfo {series and number} {NIPS'18}\ (\bibinfo
  {publisher} {Curran Associates Inc.},\ \bibinfo {year} {2018})\ pp.\ \bibinfo
  {pages} {8580--8589}\BibitemShut {NoStop}%
\bibitem [{\citenamefont {Lee}\ \emph {et~al.}(2018)\citenamefont {Lee},
  \citenamefont {Bahri}, \citenamefont {Novak}, \citenamefont {Schoenholz},
  \citenamefont {Pennington},\ and\ \citenamefont
  {Sohl-Dickstein}}]{https://doi.org/10.48550/arxiv.1711.00165}%
  \BibitemOpen
  \bibfield  {author} {\bibinfo {author} {\bibfnamefont {J.}~\bibnamefont
  {Lee}}, \bibinfo {author} {\bibfnamefont {Y.}~\bibnamefont {Bahri}}, \bibinfo
  {author} {\bibfnamefont {R.}~\bibnamefont {Novak}}, \bibinfo {author}
  {\bibfnamefont {S.~S.}\ \bibnamefont {Schoenholz}}, \bibinfo {author}
  {\bibfnamefont {J.}~\bibnamefont {Pennington}}, \ and\ \bibinfo {author}
  {\bibfnamefont {J.}~\bibnamefont {Sohl-Dickstein}},\ }in\ \href {\doibase
  10.48550/ARXIV.1711.00165} {\emph {\bibinfo {booktitle} {International
  Conference on Learning Representations}}}\ (\bibinfo {year}
  {2018})\BibitemShut {NoStop}%
\bibitem [{\citenamefont {Novak}\ \emph {et~al.}(2019)\citenamefont {Novak},
  \citenamefont {Xiao}, \citenamefont {Lee}, \citenamefont {Bahri},
  \citenamefont {Yang}, \citenamefont {Hron}, \citenamefont {Abolafia},
  \citenamefont {Pennington},\ and\ \citenamefont
  {Sohl-Dickstein}}]{https://doi.org/10.48550/arxiv.1810.05148}%
  \BibitemOpen
  \bibfield  {author} {\bibinfo {author} {\bibfnamefont {R.}~\bibnamefont
  {Novak}}, \bibinfo {author} {\bibfnamefont {L.}~\bibnamefont {Xiao}},
  \bibinfo {author} {\bibfnamefont {J.}~\bibnamefont {Lee}}, \bibinfo {author}
  {\bibfnamefont {Y.}~\bibnamefont {Bahri}}, \bibinfo {author} {\bibfnamefont
  {G.}~\bibnamefont {Yang}}, \bibinfo {author} {\bibfnamefont {J.}~\bibnamefont
  {Hron}}, \bibinfo {author} {\bibfnamefont {D.~A.}\ \bibnamefont {Abolafia}},
  \bibinfo {author} {\bibfnamefont {J.}~\bibnamefont {Pennington}}, \ and\
  \bibinfo {author} {\bibfnamefont {J.}~\bibnamefont {Sohl-Dickstein}},\ }in\
  \href {\doibase 10.48550/ARXIV.1810.05148} {\emph {\bibinfo {booktitle}
  {International Conference on Learning Representations}}}\ (\bibinfo {year}
  {2019})\BibitemShut {NoStop}%
\bibitem [{\citenamefont {Neal}(2012)}]{neal2012bayesian}%
  \BibitemOpen
  \bibfield  {author} {\bibinfo {author} {\bibfnamefont {R.~M.}\ \bibnamefont
  {Neal}},\ }\href@noop {} {\emph {\bibinfo {title} {Bayesian learning for
  neural networks}}}\ (\bibinfo  {publisher} {Springer Science \& Business
  Media},\ \bibinfo {year} {2012})\BibitemShut {NoStop}%
\bibitem [{\citenamefont {Liu}\ \emph {et~al.}(2023)\citenamefont {Liu},
  \citenamefont {Najafi}, \citenamefont {Sharma}, \citenamefont {Tacchino},
  \citenamefont {Jiang},\ and\ \citenamefont {Mezzacapo}}]{Liu_2023}%
  \BibitemOpen
  \bibfield  {author} {\bibinfo {author} {\bibfnamefont {J.}~\bibnamefont
  {Liu}}, \bibinfo {author} {\bibfnamefont {K.}~\bibnamefont {Najafi}},
  \bibinfo {author} {\bibfnamefont {K.}~\bibnamefont {Sharma}}, \bibinfo
  {author} {\bibfnamefont {F.}~\bibnamefont {Tacchino}}, \bibinfo {author}
  {\bibfnamefont {L.}~\bibnamefont {Jiang}}, \ and\ \bibinfo {author}
  {\bibfnamefont {A.}~\bibnamefont {Mezzacapo}},\ }\bibfield  {title} {\emph
  {\enquote {\bibinfo {title} {Analytic theory for the dynamics of wide quantum
  neural networks},}\ }}\bibfield  {journal} {\bibinfo  {journal} {Physical
  Review Letters}\ }\textbf {\bibinfo {volume} {130}},\ \bibinfo {pages}
  {150601} (\bibinfo {year} {2023})\BibitemShut {NoStop}%
\bibitem [{\citenamefont {Girardi}\ and\ \citenamefont
  {De~Palma}(2024)}]{https://doi.org/10.48550/arxiv.2402.08726}%
  \BibitemOpen
  \bibfield  {author} {\bibinfo {author} {\bibfnamefont {F.}~\bibnamefont
  {Girardi}}\ and\ \bibinfo {author} {\bibfnamefont {G.}~\bibnamefont
  {De~Palma}},\ }\bibfield  {title} {\emph {\enquote {\bibinfo {title} {Trained
  quantum neural networks are gaussian processes},}\ }}\bibfield  {journal}
  {\bibinfo  {journal} {arXiv preprint arXiv:402.08726}\ } (\bibinfo {year}
  {2024})\BibitemShut {NoStop}%
\bibitem [{\citenamefont {Park}\ and\ \citenamefont
  {Kastoryano}(2020)}]{Park_2020}%
  \BibitemOpen
  \bibfield  {author} {\bibinfo {author} {\bibfnamefont {C.-Y.}\ \bibnamefont
  {Park}}\ and\ \bibinfo {author} {\bibfnamefont {M.~J.}\ \bibnamefont
  {Kastoryano}},\ }\bibfield  {title} {\emph {\enquote {\bibinfo {title}
  {Geometry of learning neural quantum states},}\ }}\bibfield  {journal}
  {\bibinfo  {journal} {Physical Review Research}\ }\textbf {\bibinfo {volume}
  {2}},\ \bibinfo {pages} {023232} (\bibinfo {year} {2020})\BibitemShut
  {NoStop}%
\bibitem [{\citenamefont {Lipton}\ \emph {et~al.}(2015)\citenamefont {Lipton},
  \citenamefont {Berkowitz},\ and\ \citenamefont
  {Elkan}}]{https://doi.org/10.48550/arxiv.1506.00019}%
  \BibitemOpen
  \bibfield  {author} {\bibinfo {author} {\bibfnamefont {Z.~C.}\ \bibnamefont
  {Lipton}}, \bibinfo {author} {\bibfnamefont {J.}~\bibnamefont {Berkowitz}}, \
  and\ \bibinfo {author} {\bibfnamefont {C.}~\bibnamefont {Elkan}},\ }\bibfield
   {title} {\emph {\enquote {\bibinfo {title} {A critical review of recurrent
  neural networks for sequence learning},}\ }}\bibfield  {journal} {\bibinfo
  {journal} {arXiv preprint arXiv:1506.00019}\ } (\bibinfo {year}
  {2015})\BibitemShut {NoStop}%
\bibitem [{\citenamefont {Alzubaidi}\ \emph {et~al.}(2021)\citenamefont
  {Alzubaidi}, \citenamefont {Zhang}, \citenamefont {Humaidi}, \citenamefont
  {Al-Dujaili}, \citenamefont {Duan}, \citenamefont {Al-Shamma}, \citenamefont
  {Santamar{\'{\i}}a}, \citenamefont {Fadhel}, \citenamefont {Al-Amidie},\ and\
  \citenamefont {Farhan}}]{Alzubaidi_2021}%
  \BibitemOpen
  \bibfield  {author} {\bibinfo {author} {\bibfnamefont {L.}~\bibnamefont
  {Alzubaidi}}, \bibinfo {author} {\bibfnamefont {J.}~\bibnamefont {Zhang}},
  \bibinfo {author} {\bibfnamefont {A.~J.}\ \bibnamefont {Humaidi}}, \bibinfo
  {author} {\bibfnamefont {A.}~\bibnamefont {Al-Dujaili}}, \bibinfo {author}
  {\bibfnamefont {Y.}~\bibnamefont {Duan}}, \bibinfo {author} {\bibfnamefont
  {O.}~\bibnamefont {Al-Shamma}}, \bibinfo {author} {\bibfnamefont
  {J.}~\bibnamefont {Santamar{\'{\i}}a}}, \bibinfo {author} {\bibfnamefont
  {M.~A.}\ \bibnamefont {Fadhel}}, \bibinfo {author} {\bibfnamefont
  {M.}~\bibnamefont {Al-Amidie}}, \ and\ \bibinfo {author} {\bibfnamefont
  {L.}~\bibnamefont {Farhan}},\ }\bibfield  {title} {\emph {\enquote {\bibinfo
  {title} {Review of deep learning: concepts, cnn architectures, challenges,
  applications, future directions},}\ }}\bibfield  {journal} {\bibinfo
  {journal} {Journal of Big Data}\ }\textbf {\bibinfo {volume} {8}},\ \bibinfo
  {pages} {53} (\bibinfo {year} {2021})\BibitemShut {NoStop}%
\bibitem [{\citenamefont {Hornik}(1991)}]{HORNIK1991251}%
  \BibitemOpen
  \bibfield  {author} {\bibinfo {author} {\bibfnamefont {K.}~\bibnamefont
  {Hornik}},\ }\bibfield  {title} {\emph {\enquote {\bibinfo {title}
  {Approximation capabilities of multilayer feedforward networks},}\
  }}\bibfield  {journal} {\bibinfo  {journal} {Neural Networks}\ }\textbf
  {\bibinfo {volume} {4}},\ \bibinfo {pages} {251} (\bibinfo {year}
  {1991})\BibitemShut {NoStop}%
\bibitem [{\citenamefont {Albash}\ and\ \citenamefont
  {Lidar}(2018)}]{Albash_2018}%
  \BibitemOpen
  \bibfield  {author} {\bibinfo {author} {\bibfnamefont {T.}~\bibnamefont
  {Albash}}\ and\ \bibinfo {author} {\bibfnamefont {D.~A.}\ \bibnamefont
  {Lidar}},\ }\bibfield  {title} {\emph {\enquote {\bibinfo {title} {Adiabatic
  quantum computation},}\ }}\bibfield  {journal} {\bibinfo  {journal} {Reviews
  of Modern Physics}\ }\textbf {\bibinfo {volume} {90}},\ \bibinfo {pages}
  {015002} (\bibinfo {year} {2018})\BibitemShut {NoStop}%
\bibitem [{\citenamefont {Logan}(2005)}]{Logan_2005}%
  \BibitemOpen
  \bibfield  {author} {\bibinfo {author} {\bibfnamefont {D.~E.}\ \bibnamefont
  {Logan}},\ }\bibfield  {title} {\emph {\enquote {\bibinfo {title} {Many-body
  quantum theory in condensed matter physics{\textemdash}an introduction},}\
  }}\bibfield  {journal} {\bibinfo  {journal} {Journal of Physics A:
  Mathematical and General}\ }\textbf {\bibinfo {volume} {38}},\ \bibinfo
  {pages} {1829} (\bibinfo {year} {2005})\BibitemShut {NoStop}%
\bibitem [{\citenamefont {Klassen}\ and\ \citenamefont
  {Terhal}(2019)}]{Klassen_2019}%
  \BibitemOpen
  \bibfield  {author} {\bibinfo {author} {\bibfnamefont {J.}~\bibnamefont
  {Klassen}}\ and\ \bibinfo {author} {\bibfnamefont {B.~M.}\ \bibnamefont
  {Terhal}},\ }\bibfield  {title} {\emph {\enquote {\bibinfo {title} {Two-local
  qubit hamiltonians: when are they stoquastic?}}\ }}\bibfield  {journal}
  {\bibinfo  {journal} {Quantum}\ }\textbf {\bibinfo {volume} {3}},\ \bibinfo
  {pages} {139} (\bibinfo {year} {2019})\BibitemShut {NoStop}%
\bibitem [{\citenamefont {Austin}\ \emph {et~al.}(2011)\citenamefont {Austin},
  \citenamefont {Zubarev},\ and\ \citenamefont {Lester}}]{Austin_2011}%
  \BibitemOpen
  \bibfield  {author} {\bibinfo {author} {\bibfnamefont {B.~M.}\ \bibnamefont
  {Austin}}, \bibinfo {author} {\bibfnamefont {D.~Y.}\ \bibnamefont {Zubarev}},
  \ and\ \bibinfo {author} {\bibfnamefont {W.~A.}\ \bibnamefont {Lester}},\
  }\bibfield  {title} {\emph {\enquote {\bibinfo {title} {Quantum monte carlo
  and related approaches},}\ }}\bibfield  {journal} {\bibinfo  {journal}
  {Chemical Reviews}\ }\textbf {\bibinfo {volume} {112}},\ \bibinfo {pages}
  {263} (\bibinfo {year} {2011})\BibitemShut {NoStop}%
\bibitem [{\citenamefont {Klassen}\ \emph {et~al.}(2020)\citenamefont
  {Klassen}, \citenamefont {Marvian}, \citenamefont {Piddock}, \citenamefont
  {Ioannou}, \citenamefont {Hen},\ and\ \citenamefont
  {Terhal}}]{https://doi.org/10.48550/arxiv.1906.08800}%
  \BibitemOpen
  \bibfield  {author} {\bibinfo {author} {\bibfnamefont {J.}~\bibnamefont
  {Klassen}}, \bibinfo {author} {\bibfnamefont {M.}~\bibnamefont {Marvian}},
  \bibinfo {author} {\bibfnamefont {S.}~\bibnamefont {Piddock}}, \bibinfo
  {author} {\bibfnamefont {M.}~\bibnamefont {Ioannou}}, \bibinfo {author}
  {\bibfnamefont {I.}~\bibnamefont {Hen}}, \ and\ \bibinfo {author}
  {\bibfnamefont {B.}~\bibnamefont {Terhal}},\ }\bibfield  {title} {\emph
  {\enquote {\bibinfo {title} {Hardness and ease of curing the sign problem for
  two-local qubit hamiltonians},}\ }}\bibfield  {journal} {\bibinfo  {journal}
  {SIAM Journal on Computing}\ }\textbf {\bibinfo {volume} {49}},\ \bibinfo
  {pages} {1332} (\bibinfo {year} {2020})\BibitemShut {NoStop}%
\bibitem [{\citenamefont {Marvian}\ \emph {et~al.}(2019)\citenamefont
  {Marvian}, \citenamefont {Lidar},\ and\ \citenamefont {Hen}}]{Marvian_2019}%
  \BibitemOpen
  \bibfield  {author} {\bibinfo {author} {\bibfnamefont {M.}~\bibnamefont
  {Marvian}}, \bibinfo {author} {\bibfnamefont {D.~A.}\ \bibnamefont {Lidar}},
  \ and\ \bibinfo {author} {\bibfnamefont {I.}~\bibnamefont {Hen}},\ }\bibfield
   {title} {\emph {\enquote {\bibinfo {title} {On the computational complexity
  of curing non-stoquastic hamiltonians},}\ }}\bibfield  {journal} {\bibinfo
  {journal} {Nature Communications}\ }\textbf {\bibinfo {volume} {10}},\
  \bibinfo {pages} {1571} (\bibinfo {year} {2019})\BibitemShut {NoStop}%
\bibitem [{\citenamefont {Yang}\ and\ \citenamefont
  {Salman}(2019)}]{https://doi.org/10.48550/arxiv.1907.10599}%
  \BibitemOpen
  \bibfield  {author} {\bibinfo {author} {\bibfnamefont {G.}~\bibnamefont
  {Yang}}\ and\ \bibinfo {author} {\bibfnamefont {H.}~\bibnamefont {Salman}},\
  }\bibfield  {title} {\emph {\enquote {\bibinfo {title} {A fine-grained
  spectral perspective on neural networks},}\ }}\bibfield  {journal} {\bibinfo
  {journal} {arXiv preprint arXiv:1907.10599}\ } (\bibinfo {year}
  {2019})\BibitemShut {NoStop}%
\bibitem [{\citenamefont {Zhang}\ \emph {et~al.}(2019)\citenamefont {Zhang},
  \citenamefont {Xu}, \citenamefont {Luo},\ and\ \citenamefont
  {Ma}}]{https://doi.org/10.48550/arxiv.1905.10264}%
  \BibitemOpen
  \bibfield  {author} {\bibinfo {author} {\bibfnamefont {Y.}~\bibnamefont
  {Zhang}}, \bibinfo {author} {\bibfnamefont {Z.-Q.~J.}\ \bibnamefont {Xu}},
  \bibinfo {author} {\bibfnamefont {T.}~\bibnamefont {Luo}}, \ and\ \bibinfo
  {author} {\bibfnamefont {Z.}~\bibnamefont {Ma}},\ }\bibfield  {title} {\emph
  {\enquote {\bibinfo {title} {Explicitizing an implicit bias of the frequency
  principle in two-layer neural networks},}\ }}\bibfield  {journal} {\bibinfo
  {journal} {arXiv preprint arXiv:1905.10264}\ } (\bibinfo {year}
  {2019})\BibitemShut {NoStop}%
\bibitem [{\citenamefont {Belfer}\ \emph {et~al.}(2021)\citenamefont {Belfer},
  \citenamefont {Geifman}, \citenamefont {Galun},\ and\ \citenamefont
  {Basri}}]{https://doi.org/10.48550/arxiv.2104.03093}%
  \BibitemOpen
  \bibfield  {author} {\bibinfo {author} {\bibfnamefont {Y.}~\bibnamefont
  {Belfer}}, \bibinfo {author} {\bibfnamefont {A.}~\bibnamefont {Geifman}},
  \bibinfo {author} {\bibfnamefont {M.}~\bibnamefont {Galun}}, \ and\ \bibinfo
  {author} {\bibfnamefont {R.}~\bibnamefont {Basri}},\ }\bibfield  {title}
  {\emph {\enquote {\bibinfo {title} {Spectral analysis of the neural tangent
  kernel for deep residual networks},}\ }}\bibfield  {journal} {\bibinfo
  {journal} {arXiv preprint arXiv:2104.03093}\ } (\bibinfo {year}
  {2021})\BibitemShut {NoStop}%
\bibitem [{\citenamefont {Becca}\ and\ \citenamefont
  {Sorella}(2017)}]{becca_sorella_2017}%
  \BibitemOpen
  \bibfield  {author} {\bibinfo {author} {\bibfnamefont {F.}~\bibnamefont
  {Becca}}\ and\ \bibinfo {author} {\bibfnamefont {S.}~\bibnamefont
  {Sorella}},\ }\href {\doibase 10.1017/9781316417041} {\emph {\bibinfo {title}
  {Quantum Monte Carlo Approaches for Correlated Systems}}}\ (\bibinfo
  {publisher} {Cambridge University Press},\ \bibinfo {year}
  {2017})\BibitemShut {NoStop}%
\bibitem [{\citenamefont {Sorella}(1998)}]{Sorella:1998vk}%
  \BibitemOpen
  \bibfield  {author} {\bibinfo {author} {\bibfnamefont {S.}~\bibnamefont
  {Sorella}},\ }\bibfield  {title} {\emph {\enquote {\bibinfo {title} {Green
  function monte carlo with stochastic reconfiguration},}\ }}\bibfield
  {journal} {\bibinfo  {journal} {Physical Review Letters}\ }\textbf {\bibinfo
  {volume} {80}},\ \bibinfo {pages} {4558} (\bibinfo {year} {1998})\BibitemShut
  {NoStop}%
\bibitem [{\citenamefont
  {Yang}(2019)}]{https://doi.org/10.48550/arxiv.1902.04760}%
  \BibitemOpen
  \bibfield  {author} {\bibinfo {author} {\bibfnamefont {G.}~\bibnamefont
  {Yang}},\ }\bibfield  {title} {\emph {\enquote {\bibinfo {title} {Scaling
  limits of wide neural networks with weight sharing: Gaussian process
  behavior, gradient independence, and neural tangent kernel derivation},}\
  }}\bibfield  {journal} {\bibinfo  {journal} {arXiv preprint
  arXiv:1902.04760}\ } (\bibinfo {year} {2019})\BibitemShut {NoStop}%
\bibitem [{\citenamefont {Alemohammad}\ \emph {et~al.}(2021)\citenamefont
  {Alemohammad}, \citenamefont {Wang}, \citenamefont {Balestriero},\ and\
  \citenamefont {Baraniuk}}]{https://doi.org/10.48550/arxiv.2006.10246}%
  \BibitemOpen
  \bibfield  {author} {\bibinfo {author} {\bibfnamefont {S.}~\bibnamefont
  {Alemohammad}}, \bibinfo {author} {\bibfnamefont {Z.}~\bibnamefont {Wang}},
  \bibinfo {author} {\bibfnamefont {R.}~\bibnamefont {Balestriero}}, \ and\
  \bibinfo {author} {\bibfnamefont {R.}~\bibnamefont {Baraniuk}},\ }in\ \href
  {\doibase 10.48550/ARXIV.2006.10246} {\emph {\bibinfo {booktitle}
  {International Conference on Learning Representations}}}\ (\bibinfo {year}
  {2021})\BibitemShut {NoStop}%
\bibitem [{\citenamefont {Hron}\ \emph {et~al.}(2020)\citenamefont {Hron},
  \citenamefont {Bahri}, \citenamefont {Sohl-Dickstein},\ and\ \citenamefont
  {Novak}}]{https://doi.org/10.48550/arxiv.2006.10540}%
  \BibitemOpen
  \bibfield  {author} {\bibinfo {author} {\bibfnamefont {J.}~\bibnamefont
  {Hron}}, \bibinfo {author} {\bibfnamefont {Y.}~\bibnamefont {Bahri}},
  \bibinfo {author} {\bibfnamefont {J.}~\bibnamefont {Sohl-Dickstein}}, \ and\
  \bibinfo {author} {\bibfnamefont {R.}~\bibnamefont {Novak}},\ }in\ \href
  {\doibase 10.48550/ARXIV.2006.10540} {\emph {\bibinfo {booktitle}
  {Proceedings of the 37th International Conference on Machine Learning}}},\
  \bibinfo {series} {ICML'20}, Vol.\ \bibinfo {volume} {407}\ (\bibinfo
  {publisher} {JMLR.org},\ \bibinfo {year} {2020})\ p.~\bibinfo {pages}
  {11}\BibitemShut {NoStop}%
\bibitem [{\citenamefont
  {Yang}(2020)}]{https://doi.org/10.48550/arxiv.2006.14548}%
  \BibitemOpen
  \bibfield  {author} {\bibinfo {author} {\bibfnamefont {G.}~\bibnamefont
  {Yang}},\ }\bibfield  {title} {\emph {\enquote {\bibinfo {title} {Tensor
  programs ii: Neural tangent kernel for any architecture},}\ }}\bibfield
  {journal} {\bibinfo  {journal} {arXiv preprint arXiv:2006.14548}\ } (\bibinfo
  {year} {2020})\BibitemShut {NoStop}%
\bibitem [{\citenamefont {Chizat}\ \emph {et~al.}(2019)\citenamefont {Chizat},
  \citenamefont {Oyallon},\ and\ \citenamefont
  {Bach}}]{10.5555/3454287.3454551}%
  \BibitemOpen
  \bibfield  {author} {\bibinfo {author} {\bibfnamefont {L.}~\bibnamefont
  {Chizat}}, \bibinfo {author} {\bibfnamefont {E.}~\bibnamefont {Oyallon}}, \
  and\ \bibinfo {author} {\bibfnamefont {F.}~\bibnamefont {Bach}},\ }\enquote
  {\bibinfo {title} {On lazy training in differentiable programming},}\ in\
  \href@noop {} {\emph {\bibinfo {booktitle} {Proceedings of the 33rd
  International Conference on Neural Information Processing Systems}}}\
  (\bibinfo  {publisher} {Curran Associates Inc.},\ \bibinfo {address} {Red
  Hook, NY, USA},\ \bibinfo {year} {2019})\BibitemShut {NoStop}%
\end{thebibliography}%

\end{document}